\documentclass[11pt,a4paper]{article}
\usepackage[english]{babel}
\usepackage[utf8]{inputenc}
\usepackage[T1]{fontenc}
\usepackage{graphicx}
\usepackage{geometry}
\usepackage{tabularx}
\usepackage{psfrag}
\usepackage{fancyhdr}
\usepackage{xcolor}
\usepackage{sistyle}
\usepackage{amsfonts}
\usepackage{amsmath}
\usepackage{amssymb}
\usepackage{booktabs}
\usepackage{multirow}
\usepackage{amsthm}
\usepackage{subcaption}
\usepackage{sidecap}
\usepackage{pifont}
\usepackage[ruled]{algorithm2e}
\usepackage{enumerate}
\usepackage{enumitem}
\usepackage{bbold}
\usepackage[bottom]{footmisc}% places footnotes at page bottom
\usepackage{tikz}
	\usetikzlibrary{backgrounds,automata,shapes}
	\usetikzlibrary{decorations.pathreplacing, arrows}
\usepackage{mathrsfs}
\usepackage{xr}

\newtheorem{lemma}{Lemma}
\newtheorem{theorem}{Theorem}
\newtheorem{proposition}{Proposition}

\newtheorem*{note}{Notations}

\theoremstyle{definition}

\newtheorem{definition}{Definition}

\theoremstyle{remark}
\newtheorem*{remark}{Remark}
\newtheorem{example}{Example}
\newtheorem{hypothesis}{Assumption}

\usepackage[numbers, comma, sort&compress]{natbib}
\usepackage{subcaption}
\usepackage{enumitem}
\usepackage{bbm}
\usepackage{footmisc}
\def\1{\mathbbm{1}}
\newcommand{\Y}{\mathsf{Y}}
\newcommand{\y}{\mathsf{Y}}
\newcommand{\Yvect}{\boldsymbol{\mathsf{Y}}}
\newcommand{\yvect}{\boldsymbol{\mathsf{Y}}}

\newcommand{\kk}{\mathsf{k}}

\newcommand{\kvect}{\boldsymbol{\mathsf{k}}}

\newcommand{\Psivect}{\boldsymbol{\Psi}}
\newcommand{\Phivect}{\boldsymbol{\Phi}}
\newcommand{\X}{\mathsf{X}}
\newcommand{\Xvect}{\boldsymbol{\mathsf{X}}}
\newcommand{\mc}{\boldsymbol{\zeta}}
\newcommand{\mcc}{\zeta}
\newcommand{\U}{\mathsf{U}}

\newcommand{\alphavect}{\boldsymbol{\alpha}}

\newcommand{\paramset}{\Lambda}
\newcommand{\paramvect}{\boldsymbol{\lambda}}
\newcommand{\param}{\lambda}

\newcommand{\mmse}{\mathsf{MMSE}}
\newcommand{\vari}{\mathsf{Variance}}
\newcommand{\bias}{\mathsf{Bias}}

\newcommand{\R}{\mathsf{R}}
\newcommand{\Z}{\mathsf{Z}}

\newcommand{\Rvect}{\boldsymbol{\mathsf{R}}}
\newcommand{\Zvect}{\boldsymbol{\mathsf{Z}}}
\newcommand{\D}{\textbf{D}_2}

\graphicspath{{figures/}}

\usepackage[pdftex,pdfborder={0 0 0},
			colorlinks=true,
			linkcolor=blue,
			citecolor=green,
			urlcolor = magenta,
			pagebackref=false]{hyperref}

\begin{document}

\title{Risk Estimate under a Time-Varying Autoregressive Model for Data-Driven Reproduction Number Estimation
\thanks{B. Pascal is supported by ANR-23-CE48-0009 ``\textit{OptiMoCSI}'' and S. Vaiter by ANR-18-CE40-0005 ``\textit{GraVa}''.
The authors thank Patrice Abry for insightful discussions.}}

\author{Barbara Pascal and Samuel Vaiter
\thanks{B. Pascal is with Nantes Université, École Centrale Nantes, CNRS, LS2N, UMR 6004, F-44000 Nantes, France, (e-mail: barbara.pascal@cnrs.fr, \textit{corresponding author}); 
S. Vaiter is with CNRS, Université Côte d’Azur, LJAD, Nice,  France (\underline{e-mail:} samuel.vaiter@cnrs.fr).}}% <-this % stops a space

\maketitle

\abstract{
COVID-19 pandemic has brought to the fore epidemiological models which, though describing a wealth of behaviors, have previously received little attention in signal processing literature. In this work,  a generalized time-varying autoregressive model is considered, encompassing, but not reducing to, a state-of-the-art model of viral epidemics propagation. The time-varying parameter of this model is estimated via the minimization of a penalized likelihood estimator. A major challenge is that the estimation accuracy strongly depends on hyperparameters fine-tuning. Without available ground truth, hyperparameters are selected by minimizing specifically designed data-driven oracles, used as proxy for the estimation error. Focusing on the time-varying autoregressive Poisson model, Stein's Unbiased Risk Estimate formalism is generalized to construct asymptotically unbiased risk estimators based on the derivation of an original autoregressive counterpart of Stein's lemma. The accuracy of these oracles and of the resulting estimates are assessed through intensive Monte Carlo simulations on synthetic data. Then, elaborating on recent epidemiological models, a novel weekly scaled Poisson model is proposed, better accounting for intrinsic variability of the contaminations while being robust to reporting errors. Finally, the data-driven procedure is particularized to the estimation of COVID-19 reproduction number from weekly infection counts demonstrating its ability to tackle real-world applications.
}

\tableofcontents

\section{Introduction}
\label{sec:introduction}

\noindent \textbf{Context.} Inverse problems are ubiquitous in signal and image processing~\cite{idier2013bayesian,giovannelli2015regularization,wei2022deep,melidonis2023efficient}, with a wealth of domains of application as diverse as nonlinear physics~\cite{colas2019nonlinear,pascal2020parameter}, astronomy~\cite{denneulin2021rhapsodie}, hyperspectral imaging~\cite{borsoi2021spectral}, tomography~\cite{savanier2023deep}, cardiology~\cite{bear2018effects} and epidemiology~\cite{pascal2022nonsmooth}.
A general inverse problem consists in estimating underlying quantities of interest from direct or indirect observations.
The intricate measurement processes necessary to obtain \emph{physical}\footref{fn:physical} observations can be the source of several difficulties in estimating the quantities of interest, among which: acquisition performed in a transformed domain requiring backward transformation to access the quantity of interest~\cite{savanier2023deep}, linear or nonlinear deformation of observations through the measurement process~\cite{cachia2023fluorescence,repetti2014euclid}, and corruption by stochastic perturbations, either stemming from the physical\footnote{\label{fn:physical}\emph{Physical} is to be understood in a very broad sense, encompassing biological, epidemiological as well as human data.} phenomenon at stake or from the measurement device~\cite{giovannelli2015regularization}.
Mathematically, without loss of generality, the quantities of interest and observations can be represented by real-valued vectors $\Xvect \in \mathbb{R}^T$ and $\Yvect \in \mathbb{R}^S$ respectively, with $T, S \in \mathbb{N}^*$,  and a generic inverse problem writes
\begin{align}
\label{eq:gen-pb}
\Yvect \sim \mathcal{B}(\textbf{A}(\Xvect))
\end{align}
where $\textbf{A} : \mathbb{R}^T \rightarrow \mathbb{R}^S$ is a possibly nonlinear and non-invertible continuous transformation and $\mathcal{B}$ is a stochastic degradation, possibly \emph{data-dependent}, i.e., neither additive nor multiplicative.
For example, in low-photon imaging~\cite{melidonis2023efficient},  $\textbf{A}$ is a singular blur operator and the observations are corrupted by Poisson noise,  i.e.,  for $s \in \lbrace 1, \hdots, S\rbrace$,  the  random variable $\Y_s $ follows a Poisson distribution of intensity $(\textbf{A}\Xvect)_s$.

Most inverse problems of the form~\eqref{eq:gen-pb} are \emph{ill-posed}, e.g., when $S < T$ leading to non-injective $\textbf{A}$, hence compromising uniqueness of the quantity of interest given observations under Model~\eqref{eq:gen-pb},  or \emph{ill-conditioned}, e.g., when $\textbf{A}$ is a linear operator with a large conditioning number, inducing numerical instability which, in the presence of noise, translates into prohibitively large estimation variance.
Numerous strategies have been designed to address these challenging limitations,  including variational approaches computing a deterministic solution~\cite{Tikhonov_A_1963_j-sov-mat-dok_tikhonov_ripp,wahba1990spline,engl1996regularization,chouzenoux2015convex,foare2019semi,vacar2019unsupervised}, Bayesian techniques well-suited to probability distribution probing~\cite{tarantola2005inverse,giovannelli2015regularization,melidonis2023efficient,fort2023covid19}, and the recently developed supervised machine learning strategies relying on abundant annotated data~\cite{chan2016plug,vu2021unrolling,laumont2022bayesian}.
The present work focuses on \emph{unsupervised} variational approaches, which are particularly adapted when tackling fundamental research problems~\cite{colas2019nonlinear,pascal2020parameter} or newly emerging phenomena~\cite{pascal2022nonsmooth}, for which complex physical models have been proposed but with no annotated data available.
While the negative log-likelihood associated to Model~\eqref{eq:gen-pb} provides a measure of the fidelity of the observations to the model,  a regularization term is designed to lift the ambiguity in the solution of ill-posed and ill-conditioned problems,  e.g., leveraging negative log-prior on the quantity of interest~\cite{robert2007bayesian,gribonval2011should}.
Balancing the fidelity to the data and the regularity constraints then amounts to compute the maximum a posteriori estimate of the quantity of interest as
\begin{align}
\label{eq:MAP}
\widehat{\Xvect}(\Yvect;\paramvect) \in \underset{\Xvect \in \mathbb{R}^T}{\mathrm{Argmin}} \, \, \, \mathcal{D}(\Yvect, \textbf{A}\Xvect) + \mathcal{R}(\Xvect;\paramvect),
\end{align}
where $\mathcal{D}$ is the negative log-likelihood used as a measure of discrepancy between the observations and the model,  $\mathcal{R}$ is the negative log-prior penalizing highly irregular solutions, and $\paramvect \in \paramset$ is a set of hyperparameters controlling the level of regularity enforced in the estimate.
For example, in the study of solid friction targeting the characterization of the stick-slip regime~\cite{colas2019nonlinear,pascal2020parameter}, observations consist in a force signal measured across time corrupted by independent identically distributed (i.i.d.) additive Gaussian noise of variance $\sigma^2$. 
Hence, $T = S$ and $\textbf{A} = \textbf{I}_T$,  with the data fidelity term in~\eqref{eq:MAP} reducing to the negative log-likelihood of a Gaussian random vector of mean $\Xvect$ and scalar covariance matrix $\sigma^2 \textbf{I}_T$, that is $\mathcal{D}(\Yvect, \textbf{A}\Xvect)  = \lVert \Yvect - \Xvect \rVert_2^2/\sigma^2$.
Physicists expect the stick-slip regime to induce an almost piecewise linear force signal,  a behavior that is favored in the estimate through the penalization $\mathcal{R}(\Xvect;\param) = \param \lVert \D \Xvect \rVert_1$, where $\D : \mathbb{R}^T \rightarrow \mathbb{R}^{T-2}$ is the discrete Laplacian operator,\footnote{$\D$ consists in the discrete second order derivative defined such that $\left( \D \Xvect\right)_t = \X_{t+2} - 2 \X_{t+1} + \X_t$,   $\forall t \in \lbrace 1, \hdots, T-2\rbrace$.} the $\ell_1$-norm enforces sparsity of the second order derivative of the estimate,  and the level of sparsity is controlled by the regularization parameter $\param > 0$.
As in most ill-posed inverse problems,  the tuning of $\param$ is key to obtain an accurate estimate: for small $\param$, some noise remains, while large $\param$ might cause significant information loss due to over-regularization.

\noindent \textbf{Related works.}  Ideally, the hyperparameters would be selected by minimizing the \emph{estimation error} choosing
\begin{align}
\label{eq:true_risk}
\paramvect^\dagger \in \underset{\paramvect \in \paramset}{\mathrm{Argmin}} \, \, \left\lVert \widehat{\Xvect}(\Yvect;\paramvect) - \overline{\Xvect} \right\rVert^2_2
\end{align}
where $\overline{\Xvect}$ denotes the ground truth quantity of interest.
Though, in practice, the lack of ground truth impairs evaluation of the estimation error.
To tackle this issue, several classes of hyperparameter tuning strategies have been developed during the past decades.
For additive Gaussian noise models, leading to quadratic data-fidelity terms, regularized with quadratic penalizations, the L-curve criterion consists in selecting the regularization level by close inspection of the plot of the residual against the regularity of the solution as the regularization parameter is varied~\cite{hansen1993use}.
Not only the L-curve method is restricted to Gaussian models under Tikhonov regularization, but also it has been shown to be inaccurate when the targeted ground truth is very smooth~\cite{hanke1996limitations} and to behave inconsistently as the size of the problem increases~\cite{vogel1996non}.
An alternative way is to model the hyperparameters as random variables and to estimate them via hierarchical Bayesian techniques~\cite{robert2007bayesian,giovannelli2015regularization}. 
Such strategies come at the price of an extra complexity, as they require to specify the a priori hyperparameters distribution, and a significant computational cost, as Monte Carlo sampling is often necessary.
Recent deep learning methods take advantage of large training databases to learn the hyperparameters of variational estimators~\cite{bertocchi2020deep,nguyen2023map,gharbi2024unrolled}.
While being very accurate as soon as enough annotated data are available, these methods are not adapted to tackle fundamentally new problems for which the generation of training database would either be too costly or even not possible due to lack of available expert knowledge.
A widely used class of unsupervised methods consists in the construction of an \emph{oracle} $\mathcal{O}$ not depending explicitly on the ground truth, which can thus be evaluated in practice, and whose minimization yields an approximation of the optimal hyperparameter $\paramvect^\dagger$ defined as
\begin{align}
\label{eq:orcl_lambda}
\paramvect_{\mathcal{O}} \in \underset{\paramvect \in \paramset}{\mathrm{Argmin}} \, \, \mathcal{O}(\Yvect ; \paramvect)
\end{align}
where $\paramset \subset \mathbb{R}^L$ denotes the set of admissible hyperparameters.
For linear parametric estimates under additive Gaussian noise hypothesis, the Generalized Cross Validation strategy~\cite{golub1979generalized} consists in minimizing a data-dependent criterion constructed as the ratio between the residual sum of squares and the trace of a model-dependent linear operator.
Several extensions have been proposed to extend it to sparsity-inducing estimators~\cite{tibshirani1996regression,jansen2015generalized}.
Though, it has been shown that the Generalized Cross Validation function might not have a unique properly defined minimizer~\cite{thompson1989cautionary}, leading the Generalized Cross Validation method to fail catastrophically by producing grossly underestimated regularization parameters.
Furthermore and even more importantly, Generalized Cross Validation requires that the degradation $\mathcal{B}$ in~\eqref{eq:gen-pb} is \emph{independent} and \emph{stationary} across observations~\cite{golub1979generalized} which is a very restrictive hypothesis excluding, e.g., heteroscedastic~\cite{hooper1993iterative} and correlated noise~\cite{pascal2021automated}.
The recently proposed Approximate Leave-One-Out (ALO) Cross Validation procedure, while particularly efficient and scaling to high-dimensional problems, also relies on independence and stationarity assumptions~\cite{auddy2024approximate,nobel2024randalo}.
Among the methods relying on the design of a tractable oracle,   \emph{Stein's Unbiased Risk Estimate} based strategies, elaborating on the seminal work~\cite{stein1981estimation} to construct an approximation of the estimation risk,  were initially formulated for i.i.d. Gaussian noise model,  but have then been extended to more general noise models~\cite{eldar2008generalized}, notably including a data-dependent Poisson contribution~\cite{luisier2010image,le2014unbiased,li2017pure}.
In the past decades, Stein-based strategies have demonstrated their ability to provide accurate hyperparameter selection in numerous applications, reaching state-of-the-art performance in inverse problem resolution,  both in the variational framework, e.g.,  for multispectral image deconvolution~\cite{ammanouil2019parallel} or denoising of force signal in nonlinear physics~\cite{pascal2020parameter}, and more recently in unsupervised deep learning for image denoising~\cite{chen2022robust}.

\noindent \textbf{Contributions and outline.}
The recent COVID-19 pandemic crisis has triggered massive research efforts on epidemic modeling and surveillance, way beyond the scientific community of epidemiologists~\cite{abry2020spatial, team2020covid,flahault2020covid,oustaloup2021non,nash2022real,paireau2022ensemble,bhatia2023extending,nash2023estimating}.
Notably, the challenging estimation of the COVID-19 transmissibility in real-time, and with high accuracy despite the low quality of data collected day-by-day by health agencies, has been reformulated as a look-alike inverse problem, enabling to leverage the state-of-the-art variational estimators to get very accurate estimates of the \emph{reproduction number}, a crucial indicator quantifying the intensity of an epidemic~\cite{abry2020spatial,pascal2022nonsmooth}.
The popular model for viral epidemics proposed in~\cite{cori2013new} states that the number of new infections at time $t$, $\Z_t$, follows a Poisson distribution whose time-varying intensity is the product of the \emph{global infectiousness}, defined as a weighted sum of past infection counts $\Phi_t(\Zvect) = \sum_{s \geq 1} \varphi_s \Z_{t-s}$, and of the effective reproduction number at time $t$, $\R_t$,
\begin{align}
\label{eq:model_epi}
\Z_t \mid \Z_1, \hdots, \Z_{t-1} \sim \mathcal{P}\left( \Phi_t(\Zvect)  \R_t \right).
\end{align}
Model~\eqref{eq:model_epi} is very reminiscent of Problem~\eqref{eq:gen-pb},  the unknown quantity $\Xvect$ being $(\R_1, \hdots, \R_T)$, the role of $\textbf{A}$ being played by the linear operator $
(\R_1, \hdots, \R_T) \mapsto \left(\Phi_1(\Zvect)  \R_1, \hdots, \Phi_T(\Zvect)  \R_T\right)$
and the degradation $\mathcal{B}$ consisting in data-dependent Poisson noise.
Elaborating on this formal resemblance, the variational framework~\eqref{eq:MAP} has been fruitfully leveraged to design COVID-19 effective reproduction number $\R_t$ estimators~\cite{abry2020spatial,pascal2022nonsmooth}.
Up to now, the fine-tuning of the regularization parameters of these variational estimators has been done manually, based on expert knowledge, which not only impairs the analysis of huge amount of data, but also might reflect subjective biases of the users.
To derive a fully data-driven hyperparameter selection strategy, the main challenge lies in the design of an adapted oracle $\mathcal{O}$, for example of a Stein estimator.
Indeed,  meticulous comparison of Models~\eqref{eq:gen-pb} and~\eqref{eq:model_epi} shows that the autoregressive nature of the epidemiological model~\eqref{eq:model_epi}, which induce a dependency of the linear operator $(\R_1, \hdots, \R_T) \mapsto \left(\Phi_1(\Zvect)  \R_1, \hdots, \Phi_T(\Zvect)  \R_T\right)$ in the observation vector $\Zvect$, definitively excludes direct use of generalized Stein Unbiased Risk Estimators, requiring  $\textbf{A}$ to be statistically independent of $\Zvect$, if not deterministic~\cite{stein1981estimation,eldar2008generalized,pascal2020parameter,pascal2021automated}.
Developing a novel Stein paradigm, adapted to autoregressive models, is a challenging, though crucial, step toward the design of fully data-driven, hence objective,  strategies for reproduction number estimation.

Section~\ref{sec:model} first proposes a formal description of the generalized \emph{time-varying autoregressive models} considered, providing a general framework encompassing Model~\eqref{eq:model_epi}; then the variational framework is leveraged to design estimators of generalized time-varying autoregressive models unknown parameters under several commonly encountered noise distributions.
The proposed original autoregressive Stein paradigm is developed in Section~\ref{sec:cure}; a novel Stein's-type lemma is first derived in Section~\ref{ssec:A-lemma} for  generalized time-varying autoregressive models involving Poisson noise,  and then used to derive a prediction and an estimation unbiased risk estimators in Section~\ref{ssec:apure}; finally Finite Differences and Monte Carlo strategies are implemented to yield tractable \emph{Autoregressive Poisson Unbiased Risk Estimates}.
The accuracy of the derived risk estimates is supported by intensive numerical simulations on synthetic data, presented in Section~\ref{sec:numerical}.
Then,  in Section~\ref{sec:covid},  a novel \emph{weekly scaled Poisson} epidemiological model, accounting more precisely for the intrinsic variability of the pathogen propagation while being robust to administrative noise, is introduced.
Finally, the Autoregressive Poisson Unbiased Risk Estimate is particularized to this new model to design an original data-driven COVID-19 reproduction number estimator, which is exemplified on real data from different countries worldwide and at different pandemic stages.

\begin{note} 
$\mathbb{R}$ denotes the set of real numbers, $\mathbb{R}_+$ the nonnegative real numbers and $\mathbb{R}_+^*$ the positive real numbers. 
$\mathbb{N}$ denotes the set of nonnegative integers and $\mathbb{N}^*$ the positive integers. 
Matrices are denoted in roman bold characters, e.g., $\mathbf{L}$, 
vectors in upper case sans serif bold characters, e.g.,  $\Yvect$,  and scalars in upper case sans serif plain characters, e.g.,  $\Y$. 
To avoid unnecessary complications, deterministic and random variables are denoted in the same font; explanations are provided in case the context induces any ambiguity.
For $\yvect \in \mathbb{R}^T$ a vector of length $T$, $\mathrm{diag}(\yvect) \in \mathbb{R}^{T\times T }$ denotes the diagonal square matrix of size $T$, with diagonal coefficients consisting in the components of $\yvect$.
The entrywise product (resp. division) between two vectors is denoted by $\odot$ (resp. $\centerdot /$).
\end{note}

\section{Generalized time-varying autoregressive models}
\label{sec:model}

\subsection{Observation model}
\label{ssec:driven-AR}

The present work considers a  \emph{generalized time-varying autoregressive} mo\-del encompassing both the standard autoregressive model of finite order and the epidemiological model introduced in~\cite{cori2013new} while generalizing both of them in several directions. 
The first major originality is that the model is externally controlled by an unknown deterministic time-varying \emph{reproduction coefficient} containing the information of interest. 
This setup drastically contrasts both with the standard time series framework, mostly interested in predicting the behavior of the observations rather than inferring the coefficients~\cite{davis2021count},  and with random coefficient autoregressive models in which the time-varying coefficients are random variables~\cite{hill2014unified,regis2022random}.
Second, observations are no longer restricted to follow an independent Gaussian distribution with constant variance, but instead any distribution with prescribed mean,  possibly depending on additional time-varying parameters. 
Importantly the noise is not necessary additive but might be data-dependent or multiplicative.
And finally,  the memory term $\Psi_t$ is not necessarily a linear function, it is only assumed to be \emph{causal}, that is depending only on past observations, and \emph{smooth}. 

\begin{definition}[Generalized time-varying
autoregressive model]
Let $T \in \mathbb{N}^*$ be a time horizon,  $\overline{\Xvect} = (\overline{\X}_1, \hdots, \overline{\X}_T) \in \mathbb{R}_+^T$ a time-varying reproduction coefficient, $\Y_0 \in \mathbb{R}_+^*$ an initial state, and for each $t \in \lbrace 1, \hdots, T \rbrace$,  $\Psi_t : \mathbb{R}^{t-1} \rightarrow \mathbb{R}$ a smooth function,  with by convention $\Psi_1 = \Y_0$.
Observations $\Yvect = (\Y_1,\hdots, \Y_T)$ follow a \textit{generalized  time-varying
autoregressive} model with reproduction coefficient $\overline{\Xvect}$ and memory functions $\lbrace \Psi_t, \, t = 1, \hdots, T \rbrace$ if and only if
\begin{align}
\label{eq:model}
\forall t\in\lbrace 1, \hdots, T\rbrace, \quad \Y_t  \sim \mathcal{B}_{\alpha_t}\left( \overline{\X}_t \Psi_t(\Y_1, \hdots, \Y_{t-1}) \right), 
\end{align}
where $\mathcal{B}_{\alpha}(\mathsf{U})$ denotes a probability distribution of mean $\mathsf{U} \in \mathbb{R}$ depending on an additional parameter $\alpha\in \mathbb{R}$.
\label{def:model}
\end{definition}

The memory functions $\lbrace \Psi_t, \, t = 1, \hdots, T \rbrace$ encapsulate how the memory of past observations impacts the process  at time $t$.
In this work, they are assumed to be perfectly known.
As an external source of nonstationarity, the parameter of the probability distribution $\alpha_t$ is allowed to vary with time.

\begin{example}[Linear time-varying autoregressive model]
\label{ex:linear}
The class of \emph{linear} time-varying autoregressive models corresponds to \emph{linear} functions $\lbrace \Psi_t, \, t = 1, \hdots, T\rbrace$ defined as
\begin{align}
\label{eq:linear}
\Psi_t(\Y_1, \hdots, \Y_{t-1}) = \sum_{s = 1}^{\min(\tau,t-1)} \psi_s \Y_{t-s}
\end{align}
where $\tau \in \mathbb{N}^*$ is a finite memory horizon.
The sequence $\left\lbrace\psi_s\right\rbrace_{s=1}^\tau$, encoding the dynamical characteristics of the system, then fully characterizes the entire family of memory functions.
If, moreover,  $\left\lbrace\psi_s\right\rbrace_{s=1}^\tau$ is normalized, that is if $\sum_{s=1}^\tau \psi_s = 1$, then the global trend in the behavior of $\Y_t$ is governed by $\overline{\X}_t$: if $\overline{\X}_t >1$, $\Y_t$ is exponentially growing, while if $\overline{\X}_t<1$, $\Y_t$ decreases exponentially fast.
\end{example}

Linear time-varying autoregressive models are widely used in epidemiology,  where the sequence $\lbrace \psi_s\rbrace_{s = 1}^\tau$ corresponds to the \emph{serial interval distribution}, accounting for the randomness of the time delay between primary and secondary infections, and $\overline{\X}_t$ corresponds to the \emph{effective reproduction number}~\cite{cori2013new,abry2020spatial}. 
Figure~\ref{sfig:poiss} provides a synthetic example of linear time-varying autoregressive observations under the Poisson model~\eqref{eq:poiss_model}, with $ \psi$ corresponding to a discretized Gamma distribution of mean $6.6$ and standard deviation $3.5$ truncated at $\tau = 25$, mimicking the serial interval function of COVID-19~\cite{guzzetta2020,riccardo2020} used in~\cite{cori2013new,abry2020spatial,pascal2022nonsmooth}.
The three time periods corresponding to $\overline{\X}_t > 1$ are represented as light blue areas (first row), and results in three temporally separated bumps in $\Y_t$ (second row).

\begin{remark}[Autoregressive model of order $\tau$] The generalized time-varying
autoregressive model of Definition~\ref{def:model} encompasses the \emph{standard} autoregressive model  of order $\tau \in \mathbb{N}^*$ defined as
\begin{align}
\Y_t = \sum_{s=1}^\tau \psi_s \Y_{t-s} + \Xi_t, \quad \Xi_t \sim \mathcal{N}(0,\alpha^2)
\end{align}
where $\left( \Xi_t \right)_{t\in \mathbb{N}^*}$ is a sequence of i.i.d. Gaussian variables of  zero mean and variance $\alpha^2$, for some $\alpha > 0$~\cite[Section 2.2]{shumway2000time}.
This standard autogressive model indeed corresponds to the \emph{linear} model described in Example~\ref{ex:linear} with \emph{constant} reproduction coefficient $\overline{\X}_t = 1$, and Gaussian noise with constant variance $\alpha_t^2 = \alpha^2$.
\end{remark}

It is worth noting that the present work focuses on time-varying
autoregressive processes which are \emph{driven} by an external \emph{unknown and time-varying} reproduction coefficient, which constitutes a paradigm drastically different from the thoroughly studied autoregressive processes, preventing from using the standard tool described, e.g., in~\cite{shumway2000time,davis2021count}.

\begin{example}[Noise models]
\label{ex:noises}
Numerous probability distributions $\mathcal{B}$ are encountered in the inverse problem literature~\cite{fevotte2009nonnegative,durand2010multiplicative,jezierska2011approach,le2014unbiased,melidonis2023efficient}. 
Three major representative examples adapted to the generalized time-varying
autoregressive model introduced in Definition~\ref{def:model} are: \textit{i)} the additive Gaussian noise of variance $\alpha_t^2$
\begin{align}
\label{eq:gauss_model}
\Y_t \mid \Y_1, \hdots, \Y_{t-1}  \sim \mathcal{N}(\overline{\X}_t \Psi_t(\Yvect) ,\alpha_t^2)
\end{align}
where $\mathcal{N}(\mathsf{U} ,\alpha^2)$ denotes the Gaussian distribution of mean $\mathsf{U}$ and variance $\alpha^2$; \\
\noindent \textit{ii)} the scaled Poisson distribution with scale parameter $\alpha_t > 0$ 
\begin{align}
\label{eq:poiss_model}
\Y_t \mid \Y_1, \hdots, \Y_{t-1} \sim \alpha_t \mathcal{P} \left(\frac{ \overline{\X}_t \Psi_t(\Yvect) }{\alpha_t}  \right),
\end{align}
where $\mathcal{P}(\mathsf{U})$ denotes the Poisson distribution of intensity $\mathsf{U} $, which has mean and variance both equal to $\mathsf{U}$; \footref{fn:U}\\
\noindent \textit{iii)} the multiplicative Gamma noise of shape parameter $\alpha_t > 0$ 
\begin{align}
\label{eq:gamma_model}
\Y_t \mid \Y_1, \hdots, \Y_{t-1} \sim \mathcal{G}\left(\alpha_t, \frac{ \overline{\X}_t \Psi_t(\Yvect) }{\alpha_t}  \right),
\end{align}
where $\mathcal{G}\left(\alpha, \U  \right)$ refers to the Gamma distribution of shape parameter $\alpha$ and scale parameter $\U$.\footnote{\label{fn:U}By convention, whatever $\alpha > 0$, if $\mathsf{U} \leq 0$,  $\mathcal{P}(\mathsf{U})$ and $\mathcal{G}\left(\alpha, \U  \right)$ are Dirac distributions,  i.e., $\Y_t = 0$  deterministically.}
Both the Gaussian and Gamma distributions are well-suited to described continuous quantities.
The seminal additive Gaussian noise model is much used in biomedical signal analysis, for example to analyze electroencephalogram (EEG) recordings~\cite{zhang2010local,chang2012multivariate}, and in communications~\cite{yusuf2021autoregressive,vinogradova2022estimating}, while the multiplicative Gamma noise model has been used in financial econometrics~\cite{lanne2006mixture}.
Data-dependent Poisson noise is mostly used for count time series in which the observations can take only discrete, typically integer, values~\cite{regis2022random}, with applications in epidemiology~\cite{wallinga2004,cori2013new} and in climate studies, for example to analyze changepoints in the North Atlantic tropical cyclone record~\cite{robbins2011changepoints}.
Figure~\ref{fig:synthetic_ex} shows on the top row generalized time-varying autoregressive data under Model~\eqref{eq:model} with respectively additive Gaussian noise~\ref{sfig:gauss}, data-dependent Poisson noise~\ref{sfig:poiss}, and multiplicative Gamma noise~\ref{sfig:gamma} with prescribed time-varying reproduction coefficient, displayed in dark blue on the bottom row.
For all three examples, the function $ \psi$ shaping the memory term corresponds to a discretized Gamma distribution of mean $6.6$ and standard deviation $3.5$ truncated at $\tau = 25$, mimicking the serial interval function of COVID-19~\cite{guzzetta2020,riccardo2020}.
\end{example}

\subsection{Variational estimators}
\label{ssec:penal_like}

\begin{note}
From now on, for the sake of compactness, the driving term at $t$ will be denoted $\Psi_t(\Yvect) := \Psi_t(\Y_1, \hdots, \Y_{t-1})$ and the collection of all terms will be referred to as $\Psivect(\Yvect) = \left( \Psi_1(\Yvect), \hdots, \Psi_T(\Yvect) \right)$.
\end{note}

Given some observations $\Yvect = (\Y_1, \hdots, \Y_T)$, the most straightforward way to estimate the time-varying reproduction coefficient $\overline{\Xvect} = (\overline{\X}_1, \hdots, \overline{\X}_T)$ consists in maximizing the likelihood associated with Model~\eqref{eq:model}, yielding the Maximum Likelihood (ML) estimator:
\begin{align}
\label{eq:ML-est}
\widehat{\Xvect}^{\mathsf{ML}} = \underset{\Xvect \in \mathbb{R}^T}{\mathrm{argmin}} \, \, \mathcal{D}_{\alphavect}\left(\Yvect , \Xvect \odot\Psivect(\Yvect)\right)
\end{align}
where the discrepancy function $\mathcal{D}_{\alphavect}(\Yvect ,  \Xvect \odot \Psivect(\Yvect) )= - \ln \left( \mathbb{P}(\Yvect \lvert   \Xvect;\alphavect )\right)$ is the opposite log-likelihood of Model~\eqref{eq:model}.\footnote{Note that,  according to Model~\eqref{eq:model}, the probability distribution of the random vector $\Yvect = (\Y_1, \hdots, \Y_T)$ depends on the memory functions and on the initialization $\Y_0$ which are assumed known and deterministic, hence are not mentioned in the conditional probability. }

\begin{example}[Discrepancies, Example \ref{ex:noises} continued]
Under the additive Gaussian noise Model~\eqref{eq:gauss_model} the discrepancy is quadratic
\begin{align}
- \ln \left( \mathbb{P}(\Yvect \lvert  \Xvect ; \alphavect )\right) = \sum_{t=1}^T \frac{1}{\alpha_t^2}  \left( \Y_t -  \X_t \Psi_t(\Yvect) \right)^2 .
\end{align}
Under the scaled Poisson noise Model~\eqref{eq:poiss_model}, the discrepancy coincides with the so-called \textit{Kullback-Leibler divergence}
\begin{align}
\label{eq:DKL_gen}
- \ln \left( \mathbb{P}(\Yvect \lvert  \Xvect ; \alphavect )\right) = \sum_{t=1}^T \mathsf{d}_{\mathsf{KL}}  \left( \frac{\Y_t}{\alpha_t} \left\lvert  \frac{\X_t \Psi_t(\Yvect)}{\alpha_t} \right.\right),  \quad  \mathsf{d}_{\mathsf{KL}}(\Y\lvert \U) = \left \lbrace
\begin{array}{ll}
 \Y \ln\left( \frac{\Y}{\U}\right) + \U - \Y & \text{if } \, \Y>0, \, \U>0 \\
 \U & \text{if }\, \Y=0, \, \U \geq 0 \\
 \infty & \text{otherwise.}
\end{array}
\right.
\end{align}
Finally, under the multiplicative Gamma noise Model~\eqref{eq:gamma_model}, the discrepancy is, up to a term independent of $\U$, the Itakura-Saito divergence~\cite{fevotte2009nonnegative}
\begin{align}
- \ln \left( \mathbb{P}(\Yvect \lvert  \Xvect ; \alphavect )\right) = \sum_{t = 1}^T \mathsf{d}_{\mathsf{IS}}\left(\Y_t \lvert \alpha_t, \frac{\X_t \Psi_t(\Yvect)}{\alpha_t} \right), \quad \mathsf{d}_{\mathsf{IS}}(\Y\lvert \alpha, \U) = \left \lbrace
\begin{array}{ll}
 \frac{\Y}{\U} - \alpha \ln\left( \frac{\Y}{\U}\right) + \ln \left( \Gamma(\alpha) \Y\right) & \text{if } \, \Y>0, \, \U>0 \\
 \infty & \text{otherwise,}
\end{array}
\right.
\end{align}
where $\Gamma$ denotes the Euler gamma function.
\end{example}
For all three models, Gaussian~\eqref{eq:gauss_model}, Poisson~\eqref{eq:poiss_model} and Gamma~\eqref{eq:gamma_model}, provided that $\Psi_t(\Yvect)$ is nonzero for all  $t \in \lbrace 1, \hdots, T \rbrace$, the maximum likelihood estimator~\eqref{eq:ML-est} writes
\begin{align}
\label{eq:ML-eq}
\widehat{\X}_t^{\mathrm{ML}} = \Y_t \left/ \Psi_t(\Yvect)\right. .
\end{align}

The Maximum Likelihood estimates~\eqref{eq:ML-est} of the underlying reproduction number under all three models are displayed in light blue on the bottom row of Figure~\ref{fig:synthetic_ex}.
Due to the presence of noise in the observations (top plot,  black solid curve) the straightforward Maximum Likelihood estimate of the instantaneous reproduction coefficient (bottom plot, light blue curve) suffers from erratic fluctuations compared to ground truth (bottom plot, deep blue curve).
Such noisy estimate severely impairs the diagnostic of exponential growth based on $\X_t > 1$;  e.g.,  in Figure~\ref{sfig:poiss}, bottom plot, around $t = 50$, although ground truth is clearly above one, some values $\widehat{\X}^{\mathrm{ML}}_t < 1$ are observed.

Obtaining an accurate estimate of the reproduction coefficient thus requires to use additional information.
Widely used strategies consists in enforcing a priori constraints, such as, e.g., piecewise linearity~\cite{pascal2020parameter,abry2020spatial} and/or sparsity~\cite{pascal2022nonsmooth,gharbi2021gpu}.
To that aim, the variational framework consists in augmenting the negative log-likelihood objective of Equation~\eqref{eq:ML-est} with a regularization term enforcing a priori desirable properties on the estimate leading to parametric estimators of the form
\begin{align}
\label{eq:est_var}
\widehat{\Xvect}(\Yvect ; \paramvect) \in \underset{\Xvect \in \mathbb{R}^T}{\mathrm{Argmin}}  \, \,  \mathcal{D}_{\alphavect}\left(\Yvect ,  \Xvect \odot \Psivect(\Yvect) \right) +  \mathcal{R}( \Xvect ; \paramvect),
\end{align}
where $\paramvect = (\param_1, \hdots, \param_L) \in \paramset \subseteq \mathbb{R}_+^L$  is a vector of regularization parameters,  balancing the overall regularization level as well as the relative importance of the different constraints encoded in the penalization.
Commonly used regularization terms are composite~\cite{jin2009elastic,repetti2014euclid,repetti2021variable} and are expressed as
\begin{align}
\label{eq:regularization}
\mathcal{R}( \Xvect ; \paramvect) = \sum_{\ell = 1}^L \param_\ell \lVert \textbf{L}_{\ell} \Xvect \rVert_{q_\ell}^{q_\ell}
\end{align}
where for each $\ell \in \lbrace 1, \hdots, L\rbrace$,  $ \textbf{L}_{\ell}$ is a linear operator, $q_\ell \geq 1$ a positive exponent, and $\param_{\ell}\geq0$ is a regularization parameter balancing the importance of the $\ell$th constraint with respect to the other constraints in~\eqref{eq:regularization} and to the data-fidelity term of Equation~\eqref{eq:est_var}.
Each term of the functional enforces a specific constraint, hence enabling to take into account several regularity and sparsity properties simultaneously.
For example,  when choosing the discrete Laplacian $\textbf{L} = \textbf{D}_2$, $q_{\ell} = 1$, favors sparsity of the second order temporal derivative, and hence results in \emph{piecewise linear} estimates, while $q=2$ yields smooth estimates.
The penalized likelihood strategy sketched in Equation~\eqref{eq:est_var} is highly flexible and adapts to a large collection of noise models and constraints, hence, by favoring a priori behavior, it has the ability to provide consistent and accurate regularized estimates.
The excellent performance of variational estimators comes at the price of a cautious fine-tuning of the regularization parameters associated to each term of the penalization.
Not only is this task very tough to perform manually but also, and more importantly, in practice ground truth is not available, and it is necessary to resort to data-driven oracles to approach optimal hyperparameters.

\begin{figure}[t!]
\centering
\begin{subfigure}{0.325\linewidth}
\centering
\includegraphics[trim =  2.85cm 0mm 2.5cm 0mm, clip, width = \linewidth]{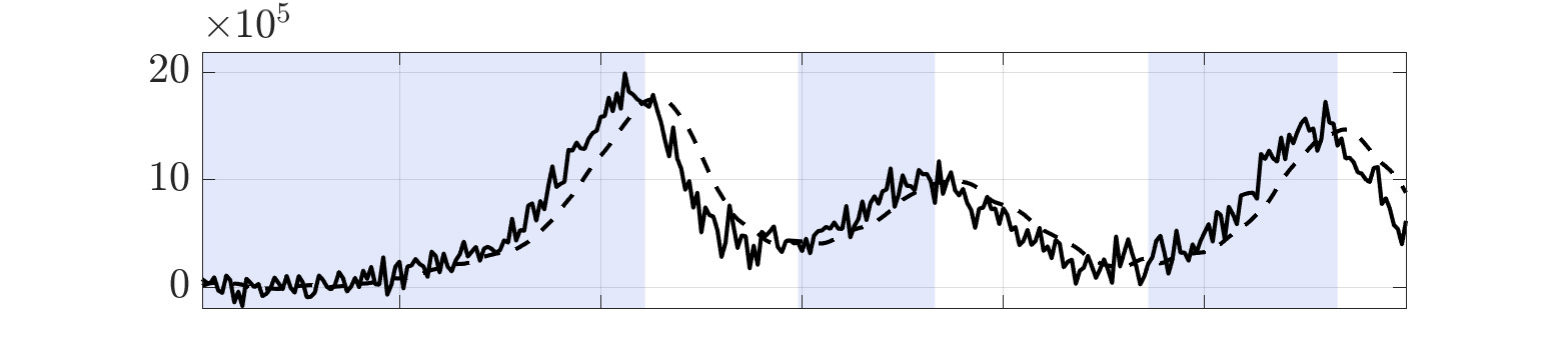}
\includegraphics[trim =  2.85cm 0mm 2.5cm 0mm, clip, width = \linewidth]{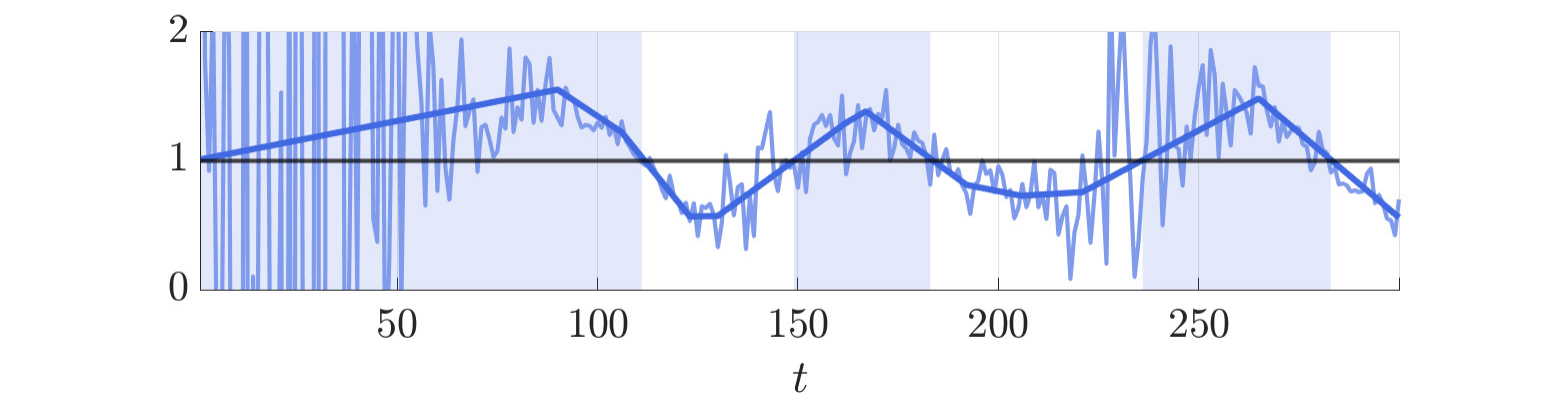}

\vspace{-2mm}
\subcaption{\label{sfig:gauss}Additive Gaussian Model~\eqref{eq:gauss_model}}
\end{subfigure}
\begin{subfigure}{0.325\linewidth}
\centering
\includegraphics[trim =  2.85cm 0mm 2.5cm 0mm, clip, width = \linewidth]{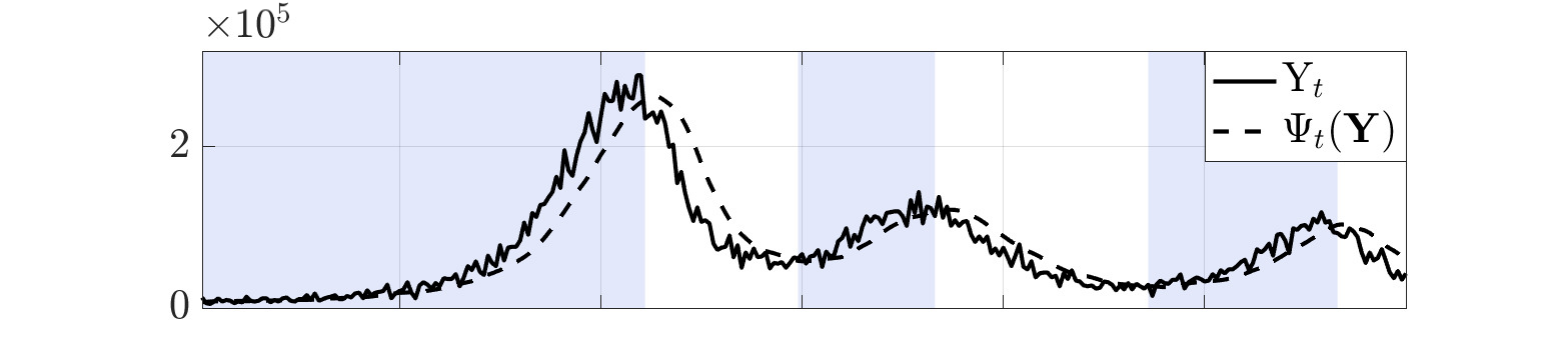}\\
\includegraphics[trim =  2.85cm 0mm 2.5cm 0mm, clip, width = \linewidth]{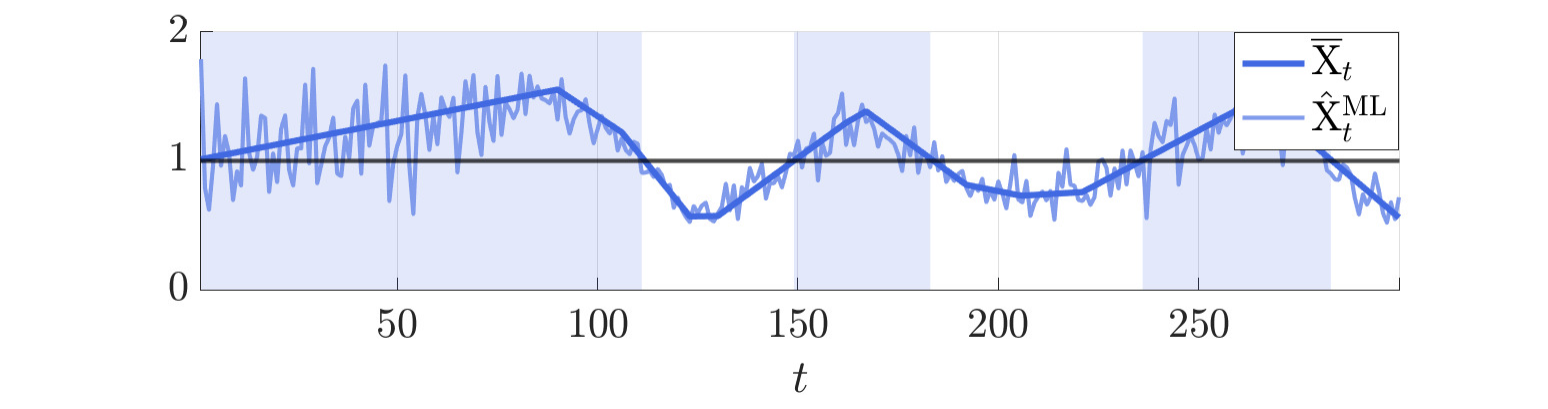}

\vspace{-2mm}
\subcaption{\label{sfig:poiss}Data-dependent Poisson Model~\eqref{eq:poiss_model}}
\end{subfigure}
\begin{subfigure}{0.325\linewidth}
\centering
\includegraphics[trim =  2.85cm 0mm 2.5cm 0mm, clip, width = \linewidth]{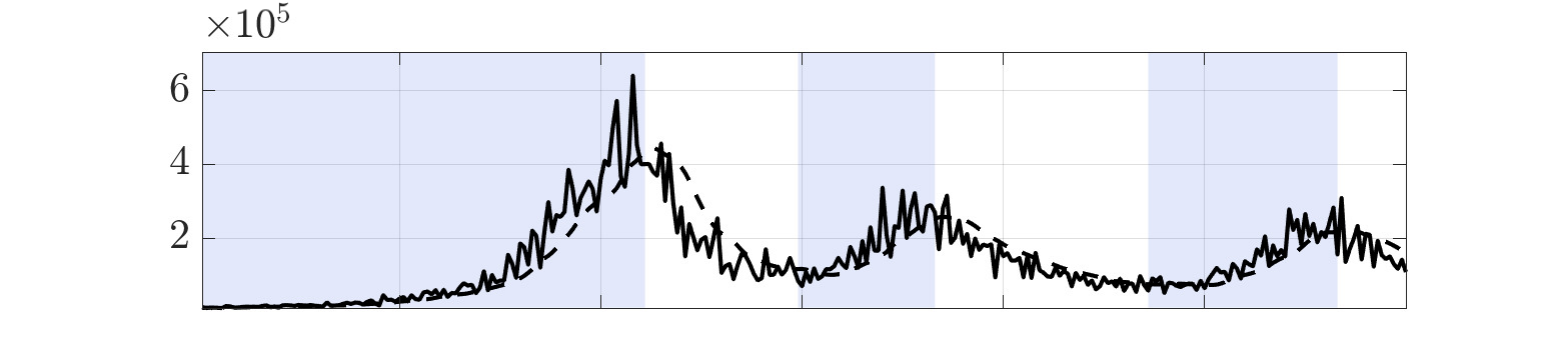}
\includegraphics[trim =  2.85cm 0mm 2.5cm 0mm, clip, width = \linewidth]{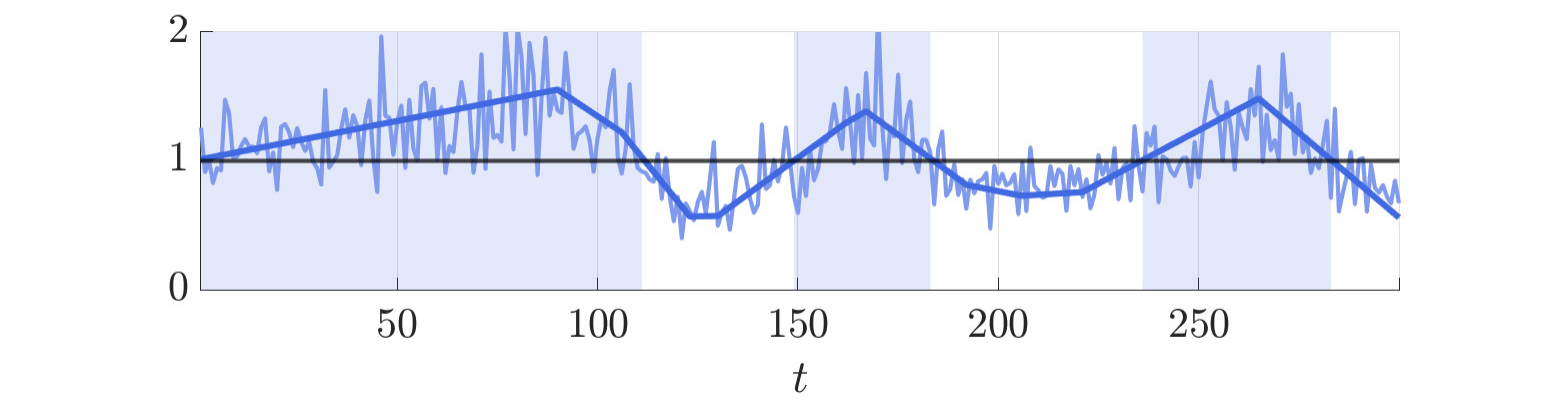}

\vspace{-2mm}
\subcaption{\label{sfig:gamma}Multiplicative Gamma Model~\eqref{eq:gamma_model}}
\end{subfigure}
\vspace{-2mm}
\caption{\textbf{Generalized time-varying autoregressive data and Maximum Likelihood reproduction coefficient estimates.}
The processes $\Y_t$, displayed as solid black curves on top row, follow Model~\eqref{eq:model} with piecewise linear reproduction coefficient $\X_t$, displayed as dark blue curve on the bottom row; linear memory functions~\eqref{eq:linear} with horizon $\tau = 25$; initialized at $\Y_0 = 10^4$ and run for $T = 300$ steps.  
The associated memory terms $\Psi_t(\Yvect)$ are displayed as dashed black curves on the top row.
Blue areas on the bottom row correspond to $\X_t \geq 1$ inducing exponential growth of $\Y_t$. 
The Maximum likelihood estimates~\eqref{eq:ML-est} of the underlying reproduction coefficient under each model are displayed as light blue curves on the bottom row.
(a) Additive Gaussian noise~\eqref{eq:gauss_model} with \emph{constant} standard deviation $\alpha_t \equiv 10^5$. 
(b) Poisson degradation~\eqref{eq:poiss_model}, with \emph{constant} scale parameter  $\alpha_t \equiv 10^3$.
(c) Gamma degradation~\eqref{eq:gamma_model}, with \emph{constant} shape parameter  $\alpha_t \equiv 20$.
\label{fig:synthetic_ex}}
\vspace{-3.5mm}
\end{figure}

\section{Unbiased risk estimators}
\label{sec:cure}

\subsection{General framework}

Given an observation model and a parametric estimator, e.g., a variational estimator of the form~\eqref{eq:est_var}, the ideal hyperparameter selection strategy would consists in minimizing the \emph{estimation} risk, defined as
\begin{align}
\label{eq:est_risk}
\mathcal{E}\left(\widehat{\Xvect}(\, \boldsymbol{\cdot} \,;\paramvect), \overline{\Xvect}\right) := \mathbb{E}_{\Yvect} \left[ \left\lVert \widehat{\Xvect}(\Yvect;\paramvect) - \overline{\Xvect}  \right\rVert^2_2\right]
\end{align}
where $\mathbb{E}_{\Yvect}$ denotes the expectation over realizations of $\Yvect$ and $\overline{\Xvect}$ the ground truth.
For inverse problems of the form~\eqref{eq:gen-pb} with an ill-conditioned or non-injective operator $\textbf{A}$,  the estimation risk is potentially numerically instable; an alternative is to shift the emphasis on the reconstruction error and
 to consider the
 \emph{prediction} risk
\begin{align}
\label{eq:pred_risk}
\mathcal{P}\left(\widehat{\Xvect}(\, \boldsymbol{\cdot}\,;\paramvect), \overline{\Xvect}\right) := \mathbb{E}_{\Yvect} \left[ \left\lVert \widehat{\Xvect}(\Yvect;\paramvect) \odot \Psivect(\Yvect) - \overline{\Xvect} \odot \Psivect(\Yvect) \right\rVert^2_2\right].
\end{align}
Both the \emph{estimation} and \emph{prediction} risks depends on the ground truth $\overline{\Xvect}$, which in practice is not available.
The purpose of this section is thus to devise \emph{oracles} for the quality of an estimate $\widehat{\Xvect}$, which are independent of the unknown ground truth, and whose  minimization provides approximately \emph{optimal} hyperparameters, where \emph{optimal} is to be understood as reaching low estimation~\eqref{eq:est_risk} or prediction~\eqref{eq:pred_risk} risk.

\subsection{A novel autoregressive Poisson lemma}
\label{ssec:A-lemma}

The remaining of the paper focuses on the Poisson model~\eqref{eq:poiss_model}, which is commonly used for modeling the pathogen spread during epidemics,  with the aim of applying the developed tools to the estimation of COVID-19 reproduction number from real infection counts in Section~\ref{sec:covid}.
It is worth noting that similar proof mechanisms could be applied to extend the analysis to Gaussian~\eqref{eq:gauss_model} and Gamma~\eqref{eq:gamma_model} noise distributions.

The cornerstone of the design of the Stein's Unbiased Risk Estimate is the seminal Stein's lemma~\cite{stein1981estimation},  turning an expectation explicitly involving ground truth, into another one in which the explicit dependency is completely removed.
Under the Poisson noise model~\eqref{eq:poiss_model}, the standard Stein's lemma cannot be used.
Further, due to the dependency of the memory term in past observations stated in Equation~\eqref{eq:model}, none of the Stein's lemmas generalized to Poisson noise~\cite{eldar2008generalized,luisier2010image,li2017pure} apply.
The major challenge in designing oracles adapted to the \emph{generalized time-varying
autoregressive Poisson} model is thus to derive a new \emph{autoregressive Poisson} Stein's lemma counterpart.

Proper definition of the estimation and prediction risks and formal derivation of the proposed \emph{autoregressive Poisson} lemma require further hypotheses: Assumption~\ref{hyp:Poisson} ensures integrability of all quantities involved, while the technical Assumption~\ref{hyp:psi-alpha},  easily checked in the practical application of Section~\ref{ssec:data-driven-Rhat},  is required to handle autoregressive models.

\begin{hypothesis}
\label{hyp:Poisson}
Let $\Theta : \mathbb{R}^T \rightarrow \mathbb{R}$, the real-valued function defined on $\mathbb{N}^T$ as
$\kvect   \mapsto  \Theta(\alpha_1 \kk_1, \hdots, \alpha_T \kk_T) \overline{\X}_t \Psi_t(\alpha_1 \kk_1, \hdots, \alpha_{t-1} \kk_{t-1}) $
is summable with respect to the autoregressive Poisson distribution obtained plugging the generalized time-varying autoregressive model~\eqref{eq:model} into the Poisson distribution~\eqref{eq:poiss_model}
\begin{align}
\label{eq:APoisson}
\mathbb{P}(\kk_1, \hdots, \kk_T) = \prod_{s = 1}^T \frac{(\overline{\X}_s \Psi_s(\yvect))^{ \kk_s}}{\kk_s!}\mathrm{e}^{-\overline{\X}_s \Psi_s(\yvect)},
\end{align}
where for all $s \in \lbrace 1, \hdots, T \rbrace$ $\y_s = \alpha_s \kk_s$.
\end{hypothesis}

\begin{hypothesis}
\label{hyp:psi-alpha}
For all $ t,s \in \lbrace 1, \hdots, T\rbrace,$
\begin{align}
\forall \yvect \in \mathbb{R}_+^T, \quad  \lvert \partial_{\y_t} \Psi_s(\yvect) \times  \alpha_t \rvert \ \ll \lvert \Psi_s(\yvect) \rvert .
\end{align}
\end{hypothesis}

It is worth insisting on the fact that, due to the dependency of the memory functions in the past observations, the components of $\Yvect$ are not independent Poisson random variables.
This is visible in Equation~\eqref{eq:APoisson} where, because $\Psi_s(\yvect)$ depends on $\y_1, \hdots, \y_{s-1}$, the right-hand side is not a separable product of independent Poisson distribution and cannot be reframed as such.
This is precisely this major difference with standard inverse problems of the form~\eqref{eq:gen-pb} which impairs the application of the standard Poisson counterpart of Stein's lemma~\cite{luisier2010image,le2014unbiased,li2017pure}.\footnote{This remark on data-dependent Poisson noise applies as well for the additive Gaussian and multiplicative Gamma noises of Example~\ref{ex:noises}.}

\begin{note} 
Let $\Theta : \mathbb{R}^T \rightarrow \mathbb{R}$ and $\alphavect \in \mathbb{R}^T$, for $t \in \lbrace 1, \hdots, T\rbrace$ the function $\Theta^{-t}$ is defined as
\begin{align}
\label{eq:def_fun_minus}
\Theta^{-t} (\yvect)= \Theta(\y_1, \hdots, \y_t - \alpha_t, \hdots, \y_T), \quad \yvect \in \mathbb{R}^T.
\end{align}
\end{note}

\begin{lemma}[Autoregressive Poisson lemma]
\label{lem:Poisson-Stein}
Let $\Yvect = (\Y_1, \hdots, \Y_T)$ observations following the generalized time-varying
autoregressive model~\eqref{eq:model} with ground truth time-varying reproduction coefficient~$\overline{\Xvect} = (\overline{\X}_1, \hdots, \overline{\X}_T) \in \mathbb{R}_+^T$ and memory functions $\Psi_t$ satisfying Assumption~\ref{hyp:psi-alpha}, corrupted by scaled Poisson noise~\eqref{eq:poiss_model} of time-varying scale parameter $\alphavect = (\alpha_1 , \hdots, \alpha_T)\in \mathbb{R}_+^T$.
Then, for $\Theta : \mathbb{R}^T \rightarrow\mathbb{R}$ satisfying Assumption~\ref{hyp:Poisson} $\forall t \in \lbrace 1, \hdots, T\rbrace$,
\begin{align}
\label{eq:poisson-lemma}
 \mathbb{E}_{\Yvect} \left[ \Theta(\Yvect) \overline{\X}_t \Psi_t(\Yvect) \right] \underset{\alphavect \rightarrow \boldsymbol{0}}{=} \mathbb{E}_{\Yvect} \left[  \Theta^{-t}(\Yvect)\Y_t\right].
\end{align}
\end{lemma}
\begin{proof}
Proof of lemma~\ref{lem:Poisson-Stein} is detailed in~\ref{ssec:proof-lemma}.
\end{proof}

Thanks to Lemma~\ref{lem:Poisson-Stein}, the ground truth-dependent expectation in the left-hand side of~\eqref{eq:poisson-lemma} is approached by a fully data-dependent one, which will turn out key in the derivation of unbiased risk estimates.

\subsection{Autoregressive Poisson Unbiased Risk Estimators}
\label{ssec:apure}

Expanding the estimation and prediction risks of Equations~\eqref{eq:est_risk} and~\eqref{eq:pred_risk} respectively, and applying Lemma~\ref{lem:Poisson-Stein} to remove the explicit dependency in the ground truth, Theorem~\ref{thm:pure-poisson} yields novel estimation and prediction Autoregressive Poisson Unbiased Risk Estimates.

\begin{theorem}
\label{thm:pure-poisson}
Let $\Yvect = (\Y_1, \hdots, \Y_T)$ be observations satisfying the requirements enunciated in Lemma~\ref{lem:Poisson-Stein}.
Let $\widehat{\Xvect}(\Yvect ;\paramvect)$ be a parametric estimator of $\overline{\Xvect}$,  such that $\forall t \in \lbrace 1, \hdots, T \rbrace$,  $\forall \paramvect \in \paramset$, $\yvect \mapsto \widehat{\X}_t(\yvect;\paramvect)$ satisfies Assumption~\ref{hyp:Poisson}.
Define the data-dependent prediction risk estimate  
\begin{align}
\begin{split}
\label{eq:apure}
\mathsf{APURE}^{\mathcal{P}}(\Yvect ; \paramvect \, \lvert \, \boldsymbol{\alpha})= \lVert  \widehat{\Xvect}(\Yvect;\paramvect) \odot \Psivect(\Yvect) \rVert^2_2    - 2 \sum_{t=1}^T \widehat{\X}_t^{-t} (\Yvect;\paramvect) \Psi_t(\Yvect) \Y_t+  \sum_{t=1}^T \left( \Y_t^2 - \alpha_t \Y_t\right)
\end{split}
\end{align}
where $\widehat{\X}_t^{-t}(\Yvect;\paramvect) = \widehat{\X}_t(\Y_1, \hdots, \Y_t - \alpha_t, \hdots, \Y_T;\paramvect)$.
Then, $\mathsf{APURE}^{\mathcal{P}}$ is an asymptotically unbiased estimate of the prediction risk in the small scale parameters limit, that is
\begin{align}
\label{eq:apure-thm1}
\mathbb{E}_{\Yvect} \left[ \mathsf{APURE}^{\mathcal{P}}(\Yvect ; \paramvect \, \lvert \, \boldsymbol{\alpha})  \right] \underset{\alphavect \rightarrow \boldsymbol{0}}{=}\mathcal{P}(\widehat{\Xvect},\overline{\Xvect}).
\end{align}
Further, assuming that $\forall t \in \lbrace 1, \hdots, T \rbrace,$  $\Psi_t(\Yvect) \neq 0$, define the data-dependent estimation risk estimate 
\begin{align}
\label{eq:apure-est}
\mathsf{APURE}^{\mathcal{E}}(\Yvect ; \paramvect \, \lvert \, \boldsymbol{\alpha})  = \lVert \widehat{\Xvect}(\Yvect;\paramvect)  \rVert^2_2
 - 2 \sum_{t=1}^T \frac{\widehat{\X}_t^{-t}(\Yvect;\paramvect)}{ \Psi_t(\Yvect)} \Y_t  + \sum_{t=1}^T \left( \frac{ \Y_t^2}{\Psi_t(\Yvect)^2} - \frac{\alpha_t \Y_t}{\Psi_t(\Yvect)^2} \right).
\end{align}
Then, $\mathsf{APURE}^{\mathcal{E}}$ is an asymptotically unbiased estimate of the prediction risk in the small scale parameters limit, that is
\begin{align}
\label{eq:apure-est-thm1}
\mathbb{E}_{\Yvect} \left[ \mathsf{APURE}^{\mathcal{E}}(\Yvect ; \paramvect \, \lvert \, \boldsymbol{\alpha})  \right] \underset{\alphavect \rightarrow \boldsymbol{0}}{=}\mathcal{E}(\widehat{\Xvect},\overline{\Xvect}).
\end{align}
\end{theorem}

\begin{proof}
Proof of Theorem~\ref{thm:pure-poisson} is developed in~\ref{ssec:apure-proof}.
\end{proof}

\subsection{Finite Difference Monte Carlo estimators}
\label{sec:fdmc}

Although fully data-driven with an explicit formula, both the estimation and prediction risk estimates $\mathsf{APURE}^{\mathcal{E}}$ and $\mathsf{APURE}^{\mathcal{P}}$ turn out to be complicated to evaluate in practice. 
Indeed, they both involve all $T$ functions $\Yvect \mapsto \widehat{\X}_t^{-t}(\Yvect)$.
Since,  in general, the estimator $\widehat{\Xvect}(\Yvect;\paramvect)$ is not separable in $t$, it is thus necessary to evaluate the estimator $T$ times.
For parametric estimators designed using the variational framework of Equation~\eqref{eq:est_var}, the involved minimization can be very costly.
Consequently, as $T$ is growing, the computational burden of the direct evaluation of $\mathsf{APURE}^{\mathcal{E}}$ and $\mathsf{APURE}^{\mathcal{P}}$ from~\eqref{eq:est_risk} and~\eqref{eq:pred_risk} rapidly becomes prohibitive.
To circumvent this difficulty,  the Finite Difference and Monte Carlo strategies~\cite{girard1989fast,ramani2008monte,deledalle2014stein,ammanouil2019parallel} are combined to yield tractable asymptotically unbiased estimation and prediction risk estimates, requiring further assumptions on the parametric estimator $\widehat{\Xvect}(\Yvect;\paramvect)$.

\begin{hypothesis}
\label{hyp:diff}
For any hyperparameters $\paramvect \in \paramset$,  the function $\yvect \mapsto \widehat{\Xvect}(\yvect ; \paramvect)$ is continuously differentiable on $\mathbb{R}_+^T$.
\end{hypothesis}

\begin{theorem}
\label{thm:fdmc}
Let $\Yvect = (\Y_1, \hdots, \Y_T)$ be observations satisfying the requirements of Lemma~\ref{lem:Poisson-Stein} and $\widehat{\Xvect}(\Yvect ; \paramvect)$ be a parametric estimator of $\overline{\Xvect}$ whose components satisfy Assumption~\ref{hyp:Poisson} as stated in Theorem~\ref{thm:pure-poisson} and satisfying Assumption~\ref{hyp:diff}.
Let $\mc \sim \mathcal{N}(\boldsymbol{0}, \textbf{I})$ a zero-mean Gaussian vector with covariance matrix the identity in dimension $T$.
Define the data-dependent Finite Difference Monte Carlo prediction risk estimate $\mathsf{APURE}_{ \mc}^{\mathcal{P}}(\Yvect ; \paramvect \, \lvert \, \boldsymbol{\alpha}) 
 =$
\begin{align}
\label{eq:fdmc-apure}
 \lVert  \widehat{\Xvect}(\Yvect;\paramvect) \odot \Psivect(\Yvect) \rVert^2_2  - 2 \sum_{t=1}^T  \widehat{\X}_t(\Yvect;\paramvect) \Psi_t(\Yvect)  \Y_t  
  + 2\left \langle   \mathrm{diag}(\boldsymbol{\alpha} \odot \Psivect(\Yvect) )  \partial_{\Yvect} \widehat{\Xvect} [\mc] , \mathrm{diag}(\Yvect)\mc \right\rangle+ \sum_{t=1}^T \left( \Y_t^2 - \alpha_t \Y_t\right)
\end{align}
where $\partial_{\Yvect} \widehat{\Xvect} [\mc]$ denotes the differential of $\Yvect \mapsto \widehat{\Xvect}(\Yvect;\paramvect)$ at the current point $(\Yvect; \paramvect)$ applied to the random vector $\mc$. Then, $\mathsf{APURE}_{\mc}^{\mathcal{P}}$ is an asymptotically unbiased estimate of the prediction risk in the small scale parameters limit:
\begin{align}
\label{eq:fdmc-pred}
\mathbb{E}_{\Yvect,\mc} \left[ \mathsf{APURE}_{\mc}^{\mathcal{P}}(\Yvect ; \paramvect \, \lvert \, \boldsymbol{\alpha})  \right] \underset{\alphavect \rightarrow \boldsymbol{0}}{=}\mathcal{P}(\widehat{\Xvect},\overline{\Xvect}).
\end{align} 
Further assuming that $\forall t \in \lbrace 1, \hdots, T \rbrace,$  $\Psi_t(\Yvect) \neq 0$, define the data-dependent Finite Difference Monte Carlo estimation risk estimate $\mathsf{APURE}_{ \mc}^{\mathcal{E}}(\Yvect ; \paramvect \, \lvert \, \boldsymbol{\alpha}) = $
\begin{align}
\label{eq:fdmc-apure-est }
 \lVert \widehat{\Xvect}(\Yvect;\paramvect)  \rVert^2_2 - 2 \sum_{t=1}^T \frac{ \widehat{\X}_t(\Yvect;\paramvect) }{ \Psi_t(\Yvect)}  \Y_t    + 2\left \langle \mathrm{diag}(\alphavect \centerdot/ \Psivect(\Yvect)) \, \partial_{\Yvect} \widehat{\Xvect} [\mc],  \mathrm{diag}(\Yvect)\mc \right\rangle  + \sum_{t=1}^T \left( \frac{ \Y_t^2}{\Psi_t(\Yvect)^2} - \frac{\alpha_t \Y_t}{\Psi_t(\Yvect)^2} \right).
\end{align}
Then, $\mathsf{APURE}_{\mc}^{\mathcal{E}}$ is an asymptotically unbiased estimate of the estimation risk in the small scale parameters limit:
\begin{align}
\label{eq:fdmc-est}
\mathbb{E}_{\Yvect,\mc} \left[ \mathsf{APURE}_{\mc}^{\mathcal{E}}(\Yvect ; \paramvect \, \lvert \, \boldsymbol{\alpha})  \right] \underset{\alphavect \rightarrow \boldsymbol{0}}{=}\mathcal{E}(\widehat{\Xvect},\overline{\Xvect}).
\end{align} 
\end{theorem}

\begin{proof}
Proof of Theorem~\ref{thm:fdmc} is developed in~\ref{sec:fdmc-proof}.
\end{proof}

Using only one realization of the Monte Carlo vector $\mc$ in the evaluation of $\mathsf{APURE}_{\mc}^{\mathcal{E}}$ and $\mathsf{APURE}_{\mc}^{\mathcal{P}}$ might lead to noisy estimates of $\mathcal{E}$ and $\mathcal{P}$; using them directly as oracles for hyperparameters selection according to~\eqref{eq:orcl_lambda} hence might result in suboptimal and unstable hyperparameter choices~\cite{lucas2022hyperparameter}.
To circumvent this issue,  it is possible to average over several independent realizations of $\mc$ to stabilize both the risk estimates and the resulting hyperparameter choice~\cite{deledalle2014stein,lucas2022hyperparameter}.

\begin{proposition}
\label{prop:robust}
Let $\Yvect = (\Y_1, \hdots, \Y_T)$ be observations satisfying the requirements enunciated in Lemma~\ref{lem:Poisson-Stein} and $\widehat{\Xvect}(\Yvect ; \paramvect)$ be a parametric estimator of $\overline{\Xvect}$  satisfying the assumptions listed in Theorem~\ref{thm:fdmc}.
Let $N \in \mathbb{N}^*$ and $(\mc^{(1)}, \hdots, \mc^{(N)})$ be independent realizations of the Monte Carlo vector $\mc \sim \mathcal{N}(\boldsymbol{0}, \textbf{I})$.
The robustified risk estimates defined as
\begin{align}
\begin{split}
\label{eq:robust}
\overline{\mathsf{APURE}_{\mc}^{\mathcal{E}}}^N = \frac{1}{N} \sum_{n = 1}^N \mathsf{APURE}_{\mc^{(n)}}^{\mathcal{E}}\quad \text{and} \quad
\overline{\mathsf{APURE}_{\mc}^{\mathcal{P}}}^N = \frac{1}{N} \sum_{n = 1}^N \mathsf{APURE}_{\mc^{(n)}}^{\mathcal{P}}
\end{split}
\end{align}
are asymptotically unbiased estimation (resp. prediction) risk estimates.
\end{proposition}

\begin{proof}
By linearity of the expectations $\mathbb{E}_{\Yvect, \mc}$ in Equation~\eqref{eq:robust} and of the limit $\alphavect \rightarrow \boldsymbol{0}$ in Equations~\eqref{eq:fdmc-pred} and \eqref{eq:fdmc-est}
\begin{align}
\begin{split}
\mathbb{E}_{\Yvect,\mc} \left[ \overline{\mathsf{APURE}_{\mc}^{\mathcal{E}}}^N \right] \underset{\alphavect \rightarrow \boldsymbol{0}}{=}\mathcal{E}(\widehat{\Xvect},\overline{\Xvect})\quad  \text{and} \quad
\mathbb{E}_{\Yvect,\mc} \left[ \overline{\mathsf{APURE}_{\mc}^{\mathcal{P}}}^N \right] \underset{\alphavect \rightarrow \boldsymbol{0}}{=}\mathcal{P}(\widehat{\Xvect},\overline{\Xvect}).
\end{split}
\end{align} 
\end{proof}

\section{Application to piecewise linear estimation}
\label{sec:numerical}

The purpose of this section is twofold: first,  to assess the ability of the robustified Finite Difference Monte Carlo risk estimates derived from Theorem~\ref{thm:fdmc} and Proposition~\ref{prop:robust},  to approximate faithfully the true estimation and prediction risks; second, to demonstrate numerically that the hyperparameters selected by minimizing these risk estimates yield accurate estimates of the reproduction coefficients from observations following the generalized time-varying autoregressive model with data-dependent Poisson noise~\eqref{eq:poiss_model}. To that aim, intensive simulations are run on synthetic data generated according to~\eqref{eq:model} and~\eqref{eq:poiss_model}.

\subsection{Synthetic data}
\label{ssec:synthetic-data}

To prepare for the application to epidemiological indicator estimation, developed in Section~\ref{sec:covid}, the ground truth reproduction coefficient $\overline{\Xvect}$ is designed piecewise linear~\cite{abry2020spatial,pascal2022nonsmooth}, imitating temporal evolution of the reproduction number of COVID-19 observed from real-world data~\cite{du2023compared}. 
All synthetic data in this section are of length $T = 300$ and share the same ground truth, represented by the deep blue curve in Figure~\ref{fig:synthetic_ex}, bottom plot, alternating expansion and recession phases,  as represented by the blue areas indicating exponential growth period characterized by $\overline{\X}_t > 1$.
The initial state is set according to the observed real COVID-19 infection counts during the imitated period at $\Y_0 = 3395$.
Mimicking the epidemiological model proposed in~\cite{cori2013new} particularized to COVID-19 pandemic, the memory functions are chosen \emph{linear}, as described in Example~\ref{ex:linear}, with a constant memory horizon of $\tau = 25$ and a sequence $\lbrace \psi_s\rbrace_{s = 1}^\tau$ chosen as the daily discretization of the serial interval distribution of COVID-19 modeled as a Gamma distribution of mean $6.6$~days and standard deviation $3.5$~days~\cite{guzzetta2020,riccardo2020}.
The scale parameter of the data-dependent Poisson noise is constant through time, i.e., $\forall t, \, \alpha_t = \alpha >0$.
Seven values of $\alpha$ logarithmically spaced between $\alpha = 10^2$, corresponding to low noise level, to a very high noise level of $\alpha = 10^{5}$,  are explored.
Figure~\ref{fig:pure} provides in top row examples of synthetic observations with the same underlying ground truth reproduction coefficient, displayed in blue in the second row plots, for three scale parameter values, corresponding,  from first to third columns,  to low $\alpha = 10^2$, medium $\alpha = 10^3$ and high $\alpha = 10^4$ noise levels.

\begin{figure}[t!]
\centering
\begin{subfigure}{0.32\linewidth}

\includegraphics[width = 0.97\linewidth]{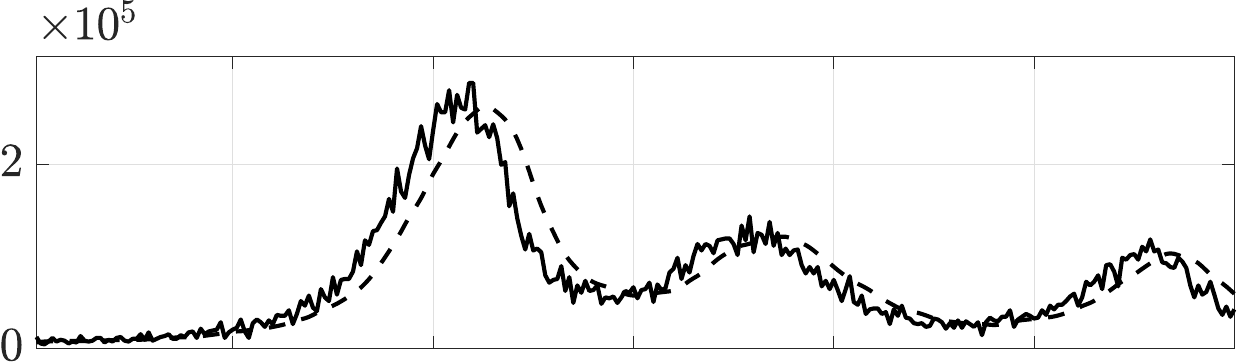}

\vspace{0.25mm}

\includegraphics[width = \linewidth]{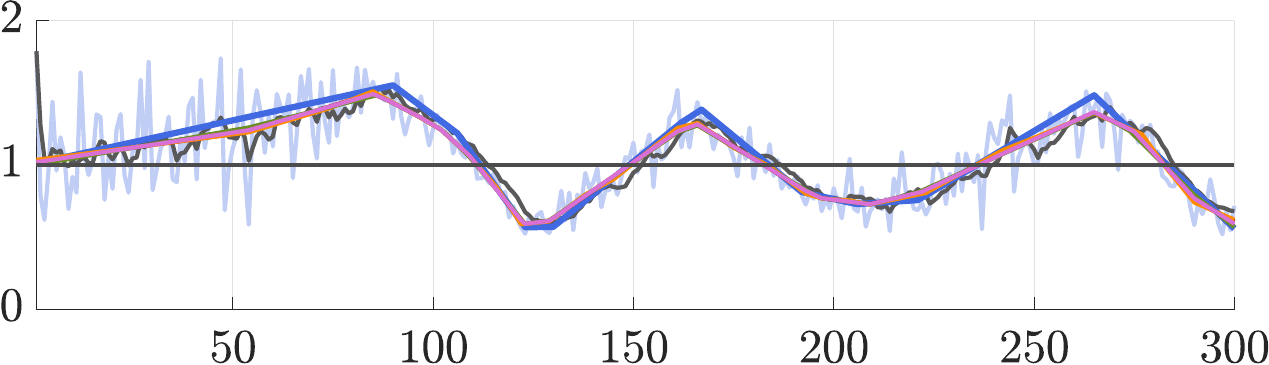}

\includegraphics[trim =  2.05cm 0mm 2.1cm 0mm, clip,width = \linewidth]{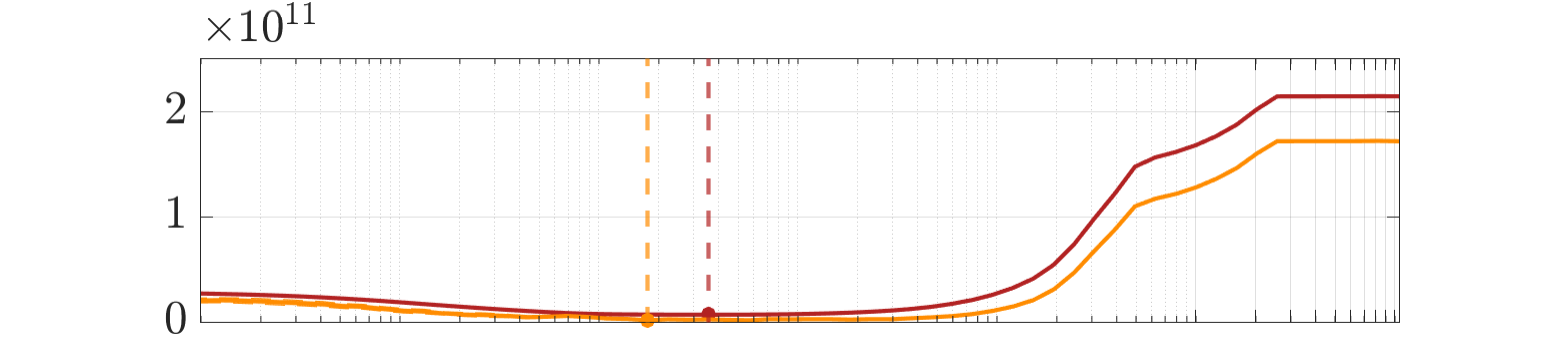}

\vspace{-1.225mm}
\includegraphics[trim =  2.05cm 0mm 2.1cm 0mm, clip,width = \linewidth]{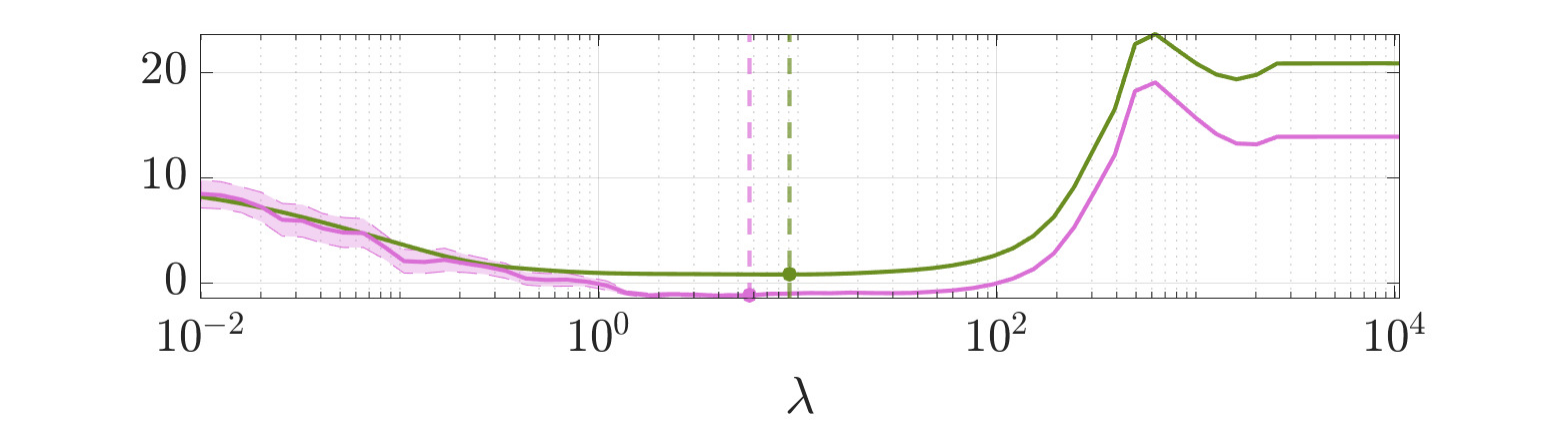}
\subcaption{$\log_{10} \alpha = 3$}
\end{subfigure}\hfill
\begin{subfigure}{0.32\linewidth}

\includegraphics[width = 0.97\linewidth]{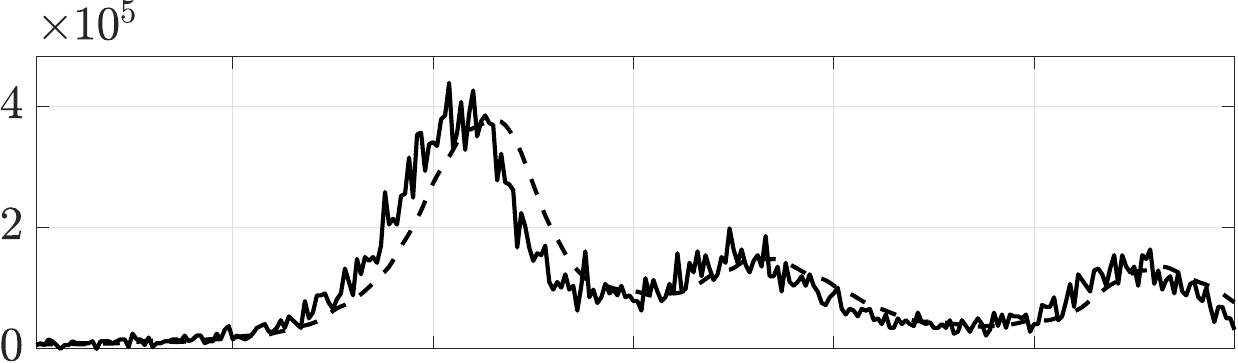}

\vspace{0.25mm}

\includegraphics[width = \linewidth]{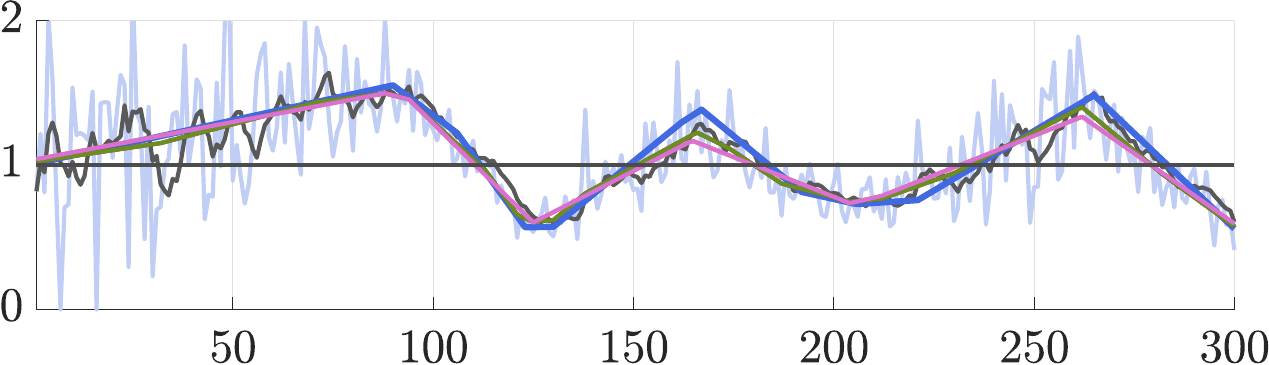}

\includegraphics[trim =  2.05cm 0mm 2.1cm 0mm, clip,width = \linewidth]{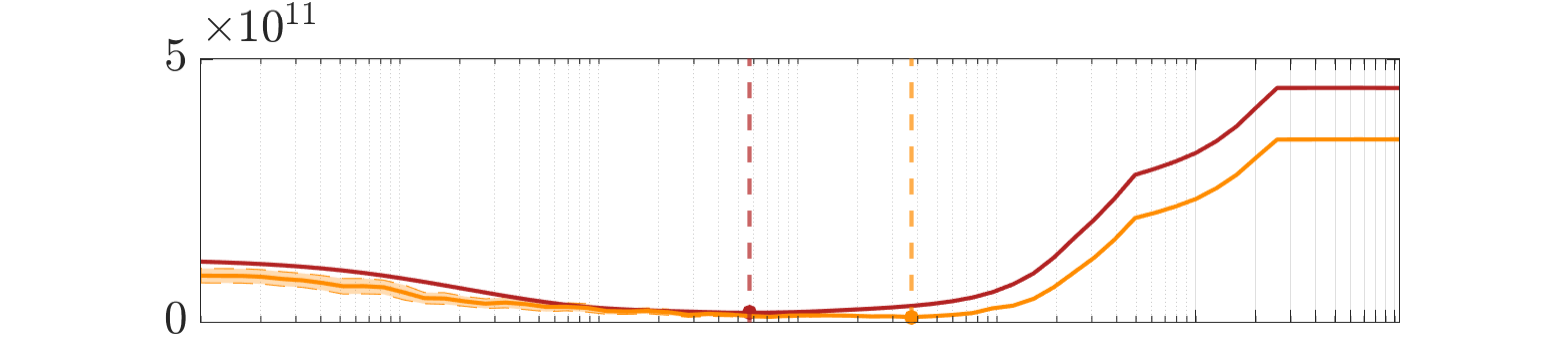}

\vspace{-1.225mm}
\includegraphics[trim =  2.05cm 0mm 2.1cm 0mm, clip,width = \linewidth]{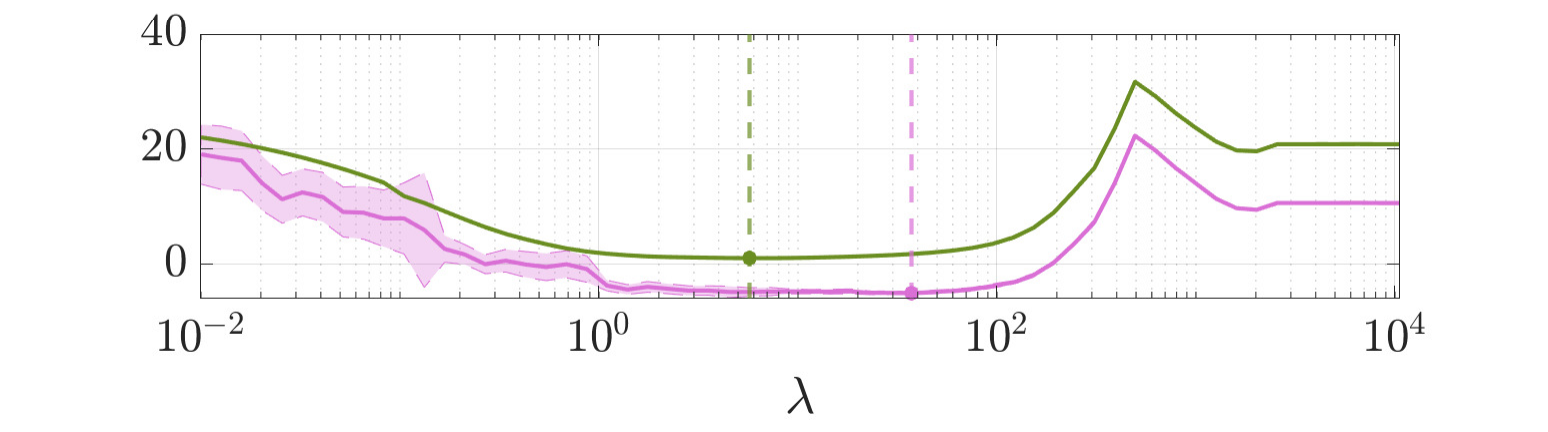}
\subcaption{$\log_{10} \alpha = 3.5$}
\end{subfigure}\hfill
\begin{subfigure}{0.32\linewidth}

\includegraphics[width = 0.97\linewidth]{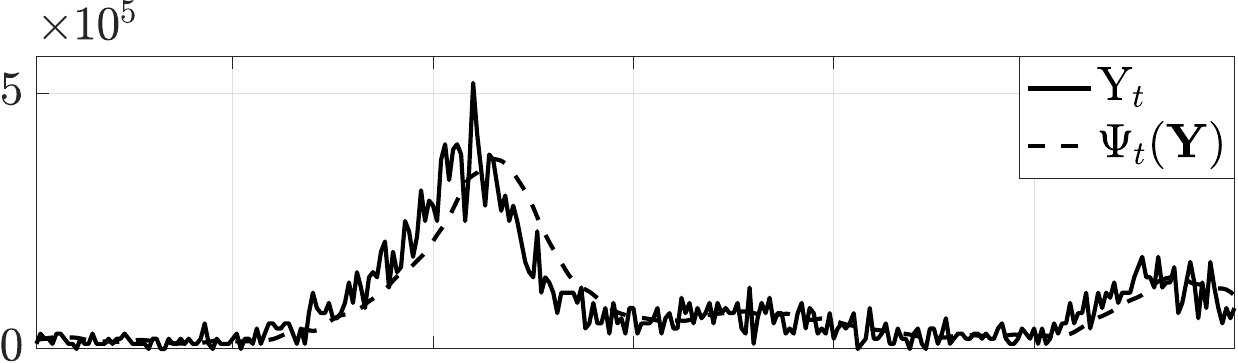}

\vspace{0.25mm}

\includegraphics[width = \linewidth]{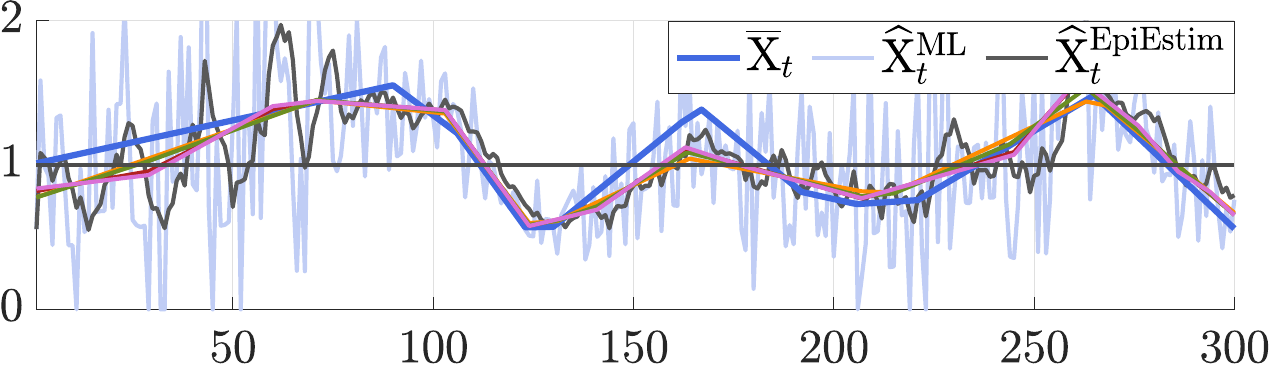}

\includegraphics[trim =  2.05cm 0mm 2.1cm 0mm, clip,width = \linewidth]{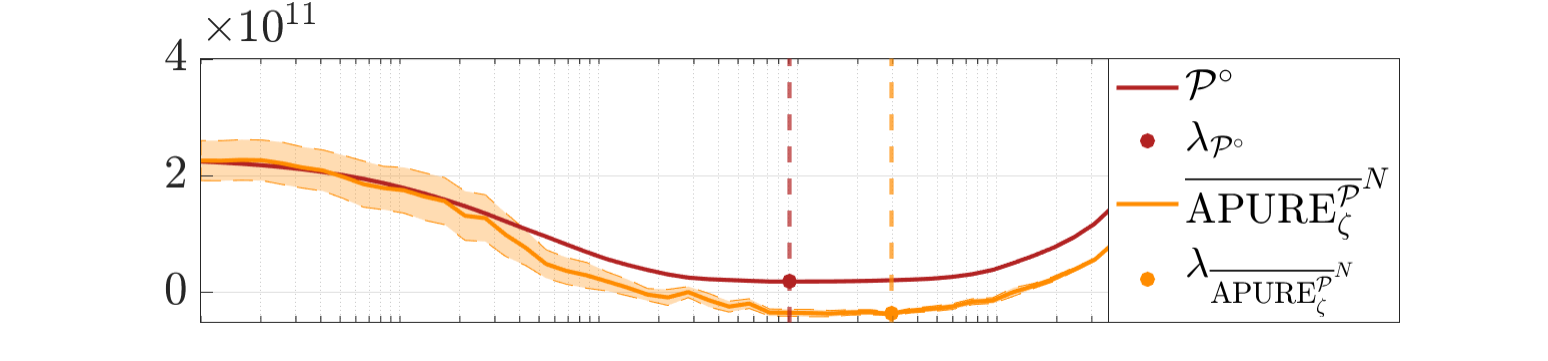}

\vspace{-1.225mm}
\includegraphics[trim =  2.05cm 0mm 2.1cm 0mm, clip,width = \linewidth]{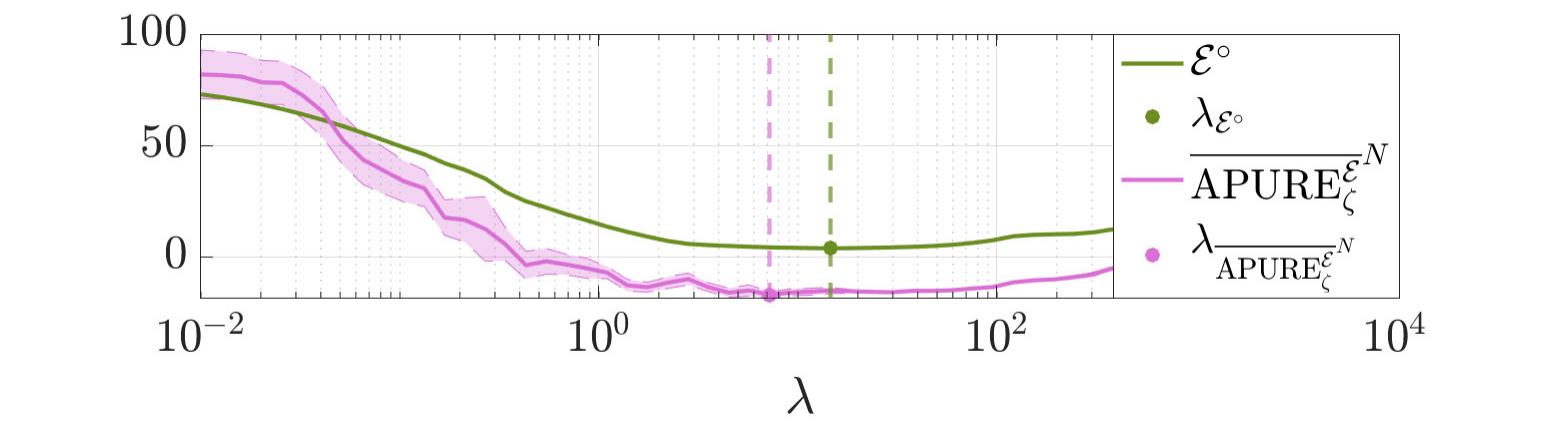}
\subcaption{$\log_{10} \alpha = 4$}
\end{subfigure}

\vspace{-2mm}
\caption{\label{fig:pure}\textbf{Estimation of the reproduction coefficient from synthetic generalized time-varying autoregressive Poisson observations using the variational estimator~\eqref{eq:practical-est} coupled with the oracle-based hyperparameter selection strategy~\eqref{eq:orcl_lambda}.}
\underline{First row:} One realization of synthetic observations drawn following the setup described in Section~\ref{ssec:synthetic-data}. 
\underline{Second row:} Underlying ground truth reproduction coefficient in blue,  Maximum Likelihood estimate~\eqref{eq:ML-eq} in light blue, state-of-the-art estimator EpiEstim~\eqref{eq:epiestim} derived for reproduction number estimation in epidemiology~\cite{cori2013new,thompson2019improved} in dark gray, and estimates obtained from the variational estimator~\eqref{eq:practical-est} combined with the oracle-based regularization parameter selection strategy, for the two ground truth-dependent oracles $\mathcal{P}^\circ$ and $\mathcal{E}^\circ$ in red and green respectively, and for the two data-driven oracles $\overline{\mathsf{APURE}^{\mathcal{P}}_{\mc}}^N$ and $\overline{\mathsf{APURE}^{\mathcal{E}}_{\mc}}^N$ in  pink and orange respectively.
\underline{Third and fourth rows:} Robustified data-driven prediction and estimation risk estimates computed on a logarithmic grid of regularization parameters $\lambda$, accompanied with their Gaussian confidence regions computed from the  $N=10$ Monte Carlo vector realizations,  and exact prediction and estimation errors; vertical dashed lines indicate optimal hyperparameters selected for each oracle.}
\vspace{-3.5mm}
\end{figure}

\subsection{Estimation strategy}
\label{ssec:est-strat}

The estimation of the  reproduction coefficient is performed through a variational procedure consisting in the minimization of the penalized negative log-likelihood functional
\begin{align}
\label{eq:practical-est}
\widehat{\Xvect}(\Yvect;\param) = \underset{\Xvect \in \mathbb{R}_+^T}{\mathrm{argmin}} \, \, \mathcal{D}_{\alphavect}\left(\Yvect , \Xvect \odot\Psivect(\Yvect)\right) + \param \lVert \textbf{D}_2 \Xvect \rVert_1
\end{align}
where $\mathcal{D}$ is the Kullback-Leibler divergence customized to the generalized time-varying autoregressive Poisson model as described in Equation~\eqref{eq:DKL_gen}, $\textbf{D}_2$ is the discrete Laplacian, which, combined with the $\ell_1$-norm favors piecewise linear behavior of the estimate, and $\lambda > 0$ is a regularization parameter balancing the data-fidelity and the penalization~\cite{abry2020spatial,pascal2022nonsmooth}.
As the minimization problem~\eqref{eq:practical-est} is convex but nonsmooth,  it is solved using the Chambolle-Pock proximal algorithm~\cite{chambolle2011first, abry2020spatial,pascal2022nonsmooth} leveraging the closed-form expression of the proximity operator of the Kullback-Leibler divergence of Equation~\eqref{eq:DKL_gen}~\cite{abry2020spatial}.

\begin{proposition}
Let $\Yvect$ observations following the generalized time-varying autoregressive Poisson model described in Section~\ref{ssec:synthetic-data}. Then, the variational estimator $\widehat{\Xvect}(\Yvect;\param)$ defined in Equation~\eqref{eq:practical-est} satisfies the assumptions enunciated in Theorems~\ref{thm:pure-poisson} and \ref{thm:fdmc}.
Hence, for any $N \in \mathbb{N}^*$, the data-driven robustified oracles $\overline{\mathsf{APURE}_{\mc}^{\mathcal{E}}}^N$ and $\overline{\mathsf{APURE}_{\mc}^{\mathcal{P}}}^N$ are asymptotically unbiased estimators of the estimation and prediction risks respectively.
\end{proposition}

Given a realization of synthetic observations $\Yvect$, the fine-tuning of the regularization parameter $\param$ is performed through the oracle minimization strategy of Equation~\eqref{eq:orcl_lambda}.
Four oracles are considered: the exact estimation and prediction errors
\begin{align}
\begin{split}
\label{eq:true-errors}
\mathcal{E}^{\circ}(\Yvect; \param) = \lVert \widehat{\Xvect}(\Yvect;\param) - \overline{\Xvect} \rVert^2_2,  \quad 
 \mathcal{P}^{\circ}(\Yvect; \param) = \lVert \widehat{\Xvect}(\Yvect;\param) \odot \Psivect(\Yvect)- \overline{\Xvect}\odot \Psivect(\Yvect) \rVert^2_2
 \end{split}
\end{align}
which, by definition of $\mathcal{E}$ and $\mathcal{P}$, are ground truth-dependent unbiased estimates of the estimation and prediction risk respectively;
and the proposed $\overline{\mathsf{APURE}_{\mc}^{\mathcal{E}}}^N$ and $\overline{\mathsf{APURE}_{\mc}^{\mathcal{P}}}^N$ introduced in Section~\ref{sec:cure}, which are fully data-driven.
Oracles are minimized through exhaustive search over a logarithmic grid of $\lambda$ from $10^{-2} \times \mathrm{std}(\Yvect)$ to $10^4\times\mathrm{std}(\Yvect)$, providing four hyperparameter choices $\param_{\mathcal{O}}$,  and resulting in four estimates $\widehat{\Xvect}(\Yvect;\param_{\mathcal{O}})$.\footnote{The dependency of $\param_{\mathcal{O}}$ in $\Yvect$ is omitted for readability.\label{fn:lambda}}
Numerical simulations aim at assessing the quality of the estimates obtained from the proposed data-driven oracles compared to estimates based on ground truth oracles, which are not usable in real-world problems, such as epidemic monitoring described in Section~\ref{sec:covid}.

Furthermore,  these four estimates are compared to the Maximum Likelihood estimate of Equation~\eqref{eq:ML-eq} and to EpiEstim\footnote{\url{https://github.com/jstockwin/EpiEstimApp}}~\cite{cori2013new,thompson2019improved}, the state-of-the-art reproduction number estimator in epidemiology, obtained as the a posteriori mean of a Bayesian model composed of the Poisson autoregressive likehood~\eqref{eq:DKL_gen} together with a conjugated gamma prior of shape parameter $a = 1$ and scale parameter $b = 5$ on the reproduction coefficient and smooth behavior constraint, yielding
\begin{align}
\label{eq:epiestim}
\forall t \in \lbrace 1, \hdots, T \rbrace, \quad \widehat{\X}_t^{\mathrm{EpiEstim}} = \frac{a +\sum_{s = 0}^{\delta-1} \Y_{t-s} }{1/b +\sum_{s = 0}^{\delta-1} \Psi_{t-s}(\Yvect)} 
\end{align}
where $\delta = 7$ mimics the counts and infectiousness averaging over seven days which enforces temporal regularity~\cite{cori2013new}.

\begin{table*}

\setlength{\tabcolsep}{0.825mm}
{\footnotesize
\hspace{-6.5mm}\begin{tabular}{clccccccc}
\toprule

Metric & \multicolumn{1}{c}{Estimator}  &  $\log_{10}\alpha = 2$ & $\log_{10}\alpha = 2.5$ & $\log_{10}\alpha = 3$ & $\log_{10}\alpha = 3.5$ & $\log_{10}\alpha = 4$ &   $\log_{10}\alpha = 4.5$     &   $\log_{10}\alpha = 5$ \\
\midrule
\multicolumn{1}{c}{$\mmse \pm \mathsf{CI} $} & Max. likelihood~\eqref{eq:ML-eq} & $1.25\pm0.08$ & $3.17\pm0.34$ & $20.58\pm9.52$ & $28.38\pm8.87$ \,& $56.25\pm14.40 $ & $83.30\pm 32.98$ & $228.00\pm374.80$\\ 
\noalign{\vskip0.75mm}
& EpiEstim~\cite{cori2013new,thompson2019improved} & $1.26\pm0.04$ & $1.47\pm0.09$& \, $3.96\pm1.33$& $5.80\pm1.93$& $8.59\pm1.64$ & $11.20\pm3.82$ \, & $18.83\pm9.18$ \, \\
\noalign{\vskip0.5mm}
& Estimation error $\mathcal{E}^{\circ}$ &  $0.51 \pm 0.02$ & $0.57\pm 0.05$ & \, $1.06\pm  0.20$ & $1.23\pm 0.40$ & $1.76 \pm 0.52$ & $1.37\pm0.26$ & $2.35 \pm1.51$ \\ 
\noalign{\vskip1.5mm}
  &Prediction error $\mathcal{P}^{\circ}$  & $0.53\pm0.02$ & $0.59\pm 0.05$ & \, $1.21\pm 0.32$ & $2.38\pm 2.00$ & $1.89\pm 0.52$ & $1.50\pm0.27$ & $3.03\pm1.84$ \\  
  &$\overline{\mathsf{APURE}^{\mathcal{P}}_{\mc}}^N$  & $0.53\pm 0.02$ & $0.59\pm 0.05$ &\,  $1.20\pm 0.33$ & $1.40\pm 0.44$ & $2.17\pm 0.68$ & $1.53\pm0.29$ &  \, $8.80\pm12.90$ \\
 &$\overline{\mathsf{APURE}^{\mathcal{E}}_{\mc}}^N$ & $0.53\pm 0.02$ & $0.59\pm 0.05$ & \, $2.03\pm 1.28$ & $4.86\pm 4.10$ & $7.42\pm 4.42$ & $23.98\pm17.82$ &  $42.08\pm44.06$ \\
\midrule

\multicolumn{1}{c}{$\bias$ -- $\vari$}  & Max. likelihood~\eqref{eq:ML-eq} & $0.52$ -- $0.73$ & $0.59$ -- $2.58$ & \, $1.69$ -- $18.88$ & \, $1.94$ -- $26.44$ & \, $3.96$ -- $52.29$ &   \,\, $4.50$ -- $78.80$& \,  $15.88$ -- $262.10$\\
\noalign{\vskip0.5mm}
& EpiEstim~\cite{cori2013new,thompson2019improved} & $1.15$ -- $0.11$ & $1.16$ -- $0.31$ & $1.33$ -- $2.63$ & $1.36$ -- $4.45$ & $1.60$ -- $6.99$ & \, $1.30$ -- $9.90$ & \,  $1.48$ -- $17.35$\\
\noalign{\vskip0.5mm}
& Estimation error $\mathcal{E}^{\circ}$ & $0.49$ -- $0.02$ &    $0.50$ -- $0.06$ &  $0.73$ -- $0.33$ &  $0.55$ -- $0.68$  &  $0.85$ -- $0.91$ &\,  $0.55$ -- $0.82$ & $0.98$ -- $1.37$  \\ 
\noalign{\vskip1mm}
&Prediction error $\mathcal{P}^{\circ}$  & $0.49$ -- $0.04$ &  $0.50$ -- $0.09$  &  $0.66$ -- $0.55$ &  $0.82$ -- $1.56$  &  $0.77$ -- $1.11$ & \, $0.48$ -- $1.02$ &  $0.65$ -- $2.38$ \\
&$\overline{\mathsf{APURE}^{\mathcal{P}}_{\mc}}^N$ & $0.50$ --  $0.03$  &   $0.53$ -- $0.07$  &  $0.70$ -- $0.50$   &  $0.56$ -- $0.84$  & $0.91$ -- $1.26$ & \, $0.57$ -- $0.96$ & $1.04$ -- $7.75$ \\
&$\overline{\mathsf{APURE}^{\mathcal{E}}_{\mc}}^N$ & $0.50$ -- $0.03$  &   $0.51$ -- $0.08$  &  $0.72$ -- $1.31$  &  $0.95$ -- $3.91$  &  $0.99$ -- $6.43$ & \, \, $1.47$ -- $22.51$ & \,  $2.38$ -- $39.70$ \\
\bottomrule
\end{tabular}}

\vspace{-1mm}
\caption{\textbf{Performance of the variational estimate~\eqref{eq:practical-est} combined with oracle-based selection of the regularization parameter compared to  maximum likelihood estimate~\eqref{eq:ML-eq} and state-of-the-art EpiEstim estimator~\cite{cori2013new,thompson2019improved}.}
Ground truth-dependent oracles,  $\mathcal{E}^\circ$ and $\mathcal{P}^\circ$ defined in Equation~\eqref{eq:true-errors}, are compared with the proposed fully data-driven robustified $\overline{\mathsf{APURE}^{\mathcal{E}}_{\mc}}^N$ and $\overline{\mathsf{APURE}^{\mathcal{P}}_{\mc}}^N$ derived in Proposition~\eqref{prop:robust}, averaged over $N=10$ realizations of the Monte Carlo vector for seven logarithmically spaced scale parameters $\alpha$.
\underline{Second to seventh rows:} Minimal Mean Squared Error accompanied with 95$\%$ Gaussian confidence intervals.  \underline{Eighth to eleventh rows:} Estimation bias and variance.
Performances are computed on $Q=20$ realizations of synthetic observations drawn following the setup described in Section~\ref{ssec:synthetic-data}.
\label{tab:pure}}
\vspace{-3.5mm}
\end{table*}

\subsection{Performance evaluation}

For each scale parameter explored,  performances are evaluated on a collection of $Q = 20$ independent synthetic observations $\lbrace \Yvect^{(q)}, \, q = 1, \hdots, Q \rbrace$ generated following the setup of Section~\ref{ssec:synthetic-data}.  
For $\mathcal{O} \in \lbrace\mathcal{E}^\circ,\mathcal{P}^\circ,\overline{\mathsf{APURE}_{\mc}^{\mathcal{E}}}^N,\overline{\mathsf{APURE}_{\mc}^{\mathcal{P}}}^N\rbrace$, the hyperparameters selected through the oracle strategy applied to the $q$th realization are denoted\footref{fn:lambda}
\begin{align}
\param_{\mathcal{O}}^{(q)} \in \underset{\param \in \mathbb{R_+}}{\mathrm{Argmin}} \, \, \mathcal{O}(\Yvect^{(q)};\param).
\end{align}
Several performance criteria are considered.
First, the \emph{Minimal Mean Squared Error} over the $Q$ realizations:
\begin{align}
\label{eq:mmse}
\mmse = \frac{1}{Q} \sum_{q = 1}^Q \left\lVert \widehat{\Xvect}\left(\Yvect^{(q)};\param_\mathcal{O}^{(q)}\right) - \overline{\Xvect} \right\rVert^2_2,
\end{align}
quantifies the overall accuracy of the estimate obtained when selecting hyperparameters so as to minimize the oracle $\mathcal{O}$.
It is accompanied by its 95$\%$ Gaussian confidence interval:
\begin{align}
\label{eq:gci}
\mathsf{CI} =\frac{1.96}{ \sqrt{Q}} \times \frac{1}{Q}  \sum_{q = 1}^Q \left(\left\lVert \widehat{\Xvect}\left(\Yvect^{(q)};\param_\mathcal{O}^{(q)}\right) - \overline{\Xvect} \right\rVert^2_2 - \mmse\right)^2.
\end{align}
Let $\langle \widehat{\Xvect} \rangle_\mathcal{O}$ the mean estimate obtained using the oracle $\mathcal{O}$ over the $Q$ realizations,\footnote{The mean estimate writes $\displaystyle
\langle \widehat{\Xvect} \rangle_\mathcal{O} =  \frac{1}{Q} \sum_{q = 1}^Q \widehat{\Xvect}\left(\Yvect^{q)};\param_\mathcal{O}^{(q)}\right).$} then the Minimal Mean Squared Error further decomposes into a squared bias term and a variance term, $\mmse = \bias + \vari$, with
\begin{align}
\bias =  \left\lVert \langle \widehat{\Xvect} \rangle_\mathcal{O} - \overline{\Xvect} \right\rVert^2_2,  \quad
\vari =  \frac{1}{Q} \sum_{q = 1}^Q \left\lVert \widehat{\Xvect}\left(\Yvect^{(q)};\param_\mathcal{O}^{(q)}\right) - \langle \widehat{\Xvect} \rangle_\mathcal{O} \right\rVert^2
\end{align}
reported together with the $\mmse$ for further comparison.

\begin{remark}
The number of realizations $Q$ over which  performance  are computed has been chosen after observing that all aforementioned metrics together with their confidence intervals had stabilized after $Q = 20$ realizations.
\end{remark}

\subsection{Results}

Figure~2 compares the four proposed oracle-based estimation strategies, the Maximum Likelihood~\eqref{eq:ML-eq} and the state-of-the-art EpiEstim~\cite{cori2013new,thompson2019improved} estimators on one realization of generalized time-varying autoregressive Poisson synthetic observations for three values of the scale parameter, corresponding to low, medium and high noise levels.
Synthetic observations and memory terms are displayed in the first row; the underlying ground truth reproduction coefficient and its estimates obtained using the four different oracles, Maximum Likelihood strategy and EpiEstim estimator are plotted on the second row;  third (resp. fourth) row compares the prediction (resp. estimation) ground truth and data-dependent oracles, and the associated optimal hyperparameters.

First of all,  for all values of the scale parameter $\alpha$, both the Maximum likelihood and EpiEstim estimates achieve poor accuracy compared to the four oracle-based regularized estimates.
Second, for the three noise levels,  both the prediction and estimation data-dependent oracles, displayed as the orange curves in the third row plots and pink curves in the fourth row plots respectively, approximate very closely the true prediction and estimation errors, displayed as the red curves in the third row plots and green curves in the fourth row plots respectively, throughout the large range of regularization parameter explored.
For low and medium noise levels, first and second columns in Figure~\ref{fig:pure}, the four oracles are flat enough in the optimal regularization parameter region so that the small differences between the selected hyperparameters are not visible on the obtained estimates, displayed in the second row, which are all of equal quality.
When the noise level gets higher, all oracles are more picked around their minima and the hyperparameters selected by ground truth dependent or data-driven oracles are very close, yielding similarly good estimates.
As expected, the 95$\%$ Gaussian confidence regions around the robustified data-driven oracles,  computed from the $N = 10$ realizations of the Monte Carlo vector and displayed in deemed colors, thicken as the noise level increases, while remaining of reasonable size, demonstrating the stability of the parameter selection strategy relying on the data-driven oracles.

Table~\ref{tab:pure} provides systematic performances computed on $Q = 20$ realizations of the observations for each of the seven noise levels explored.
In agreement with Figure~\ref{fig:pure}, Maximum Likelihood and EpiEstim performance lie far behind those of the proposed oracle-based regularized strategies.
As expected, the larger $\alpha$, the larger the Minimal Mean Squared Error, consistently for all six estimators. 
Though, the estimation accuracy when using the two ground truth-dependent oracles,  fourth and fifth rows, and the robustified \emph{prediction} unbiased risk estimate, sixth row, increases very slowly with $\alpha$.
Furthermore, at fixed noise level, using either $\mathcal{E}^\circ$, $\mathcal{P}^\circ$ or  $\overline{\mathsf{APURE}^{\mathcal{P}}_{\mc}}^N$ leads to equivalent Minimal Mean Squared Error.
This shows, first, that the estimation accuracy is unaltered when replacing the exact \emph{estimation} error $\mathcal{E}^\circ$ by the exact \emph{prediction} error $\mathcal{P}^\circ$, and second, that the data-driven oracle $\overline{\mathsf{APURE}^{\mathcal{P}}_{\mc}}^N$ yields estimates of very similar quality.
Considering the performance of the robustified unbiased \emph{estimation} risk estimate, reported in the seventh row,  they remain similarly good as the three other oracles for low to medium noise levels, but then suddenly drop for $\alpha > 10^3$ while the associated 95$\%$ Gaussian confidence intervals are also thickening violently.
This shows that, despite the robustification procedure of Proposition~\ref{prop:robust}, the data-driven oracle $\overline{\mathsf{APURE}^{\mathcal{E}}_{\mc}}^N$ is unstable, which is probably due to the division by $\Psi_t(\Yvect)$ which order of magnitude varies significantly with $t$, between $10^3$ and  $10^{5}$ in these examples, as can be observed on Figure~\ref{fig:pure}, top row.
To gain further insight, the decomposition of the Minimal Mean Squared Error into the squared bias and the variance is reported in eighth to thirteenth rows of Table~\ref{tab:pure}. 
For low noise levels $\alpha < 10^3$,  third and fourth columns, the bias, which is intrinsic to all regularized estimation strategies of the form~\eqref{eq:MAP}, is responsible for almost all the estimation error.
When the noise level exceeds $\alpha = 10^3$, squared bias and variance contribute equally to the estimation error, as observed in fifth to eighth columns.
These results advocate, in a practical context, to use preferably the robustified unbiased \emph{prediction} risk estimate which appears both very accurate in the selection of the optimal regularization parameter and more robust to medium to high noise levels.
Finally, it is worth noting that for very high noise level $\alpha = 10^5$,  corresponding to the ninth column of Table~\ref{tab:pure}, even the most robust oracle-based estimators are impaired, both because the estimation task is extremely difficult and because such values of the scale parameter do not permit to satisfy Assumption~\ref{hyp:psi-alpha}.

\section{Application to epidemiology}
\label{sec:covid}
A major motivation of the present work is the need for data-driven hyperparameter fine-tuning strategies for recently proposed COVID-19 reproduction estimators leveraging the variational framework sketched in Equation~\eqref{eq:MAP}~\cite{abry2020spatial, pascal2022nonsmooth,du2023compared}.

\subsection{Weekly scaled Poisson epidemiological model}
\label{ssec:weekly-model}

The considered epidemiological model, briefly introduction in Section~\ref{sec:introduction}, in  Equation~\eqref{eq:model_epi},  was initially proposed in~\cite{cori2013new} and states that, conditionally to past infection counts $\Z_1, \hdots, \Z_{t-1}$, the number of new infections at time $t$, denoted $\Z_t$, follows a Poisson distribution of intensity equal to the product of the effective reproduction number at time $t$, $\R_t$, and the global infectiousness $\Phi_t(\Zvect) = \sum_{s \geq 1} \varphi_s \Z_{t-s}$ with $\varphi$ the serial interval distribution.\footnote{The \emph{serial interval} is the random delay between primary and secondary infections; its distribution encodes typical propagation time scale.}
Although very accurate for a posteriori consolidated data~\cite{cori2013new}, the daily version of this model appeared not suitable for real-time COVID-19 daily infection counts, an example of which is displayed in Figure~\ref{fig:daily}, gray curve.
Indeed, COVID-19 data are severely corrupted by administrative noise,  taking the form of missing counts during week-ends and holidays,  large cumulative counts on Mondays, pseudo-seasonalities, erroneous samples, to name but a few.
This noise has two major effects: first, reported COVID-19 daily infections counts contain outlier values,  notably during  week-ends as illustrated  in the gray curve in Figure~\ref{fig:daily}; second, the variance of the reported counts reflects both the intrinsic variance of the propagation process and the additional variance induced by the fluctuating administrative delays and errors.
To circumvent the presence of outlier samples, a classical strategy in epidemiology is to consider aggregated data~\cite{ferguson2016countering,charniga2021spatial,nash2023estimating}, e.g., at the scale of the week as illustrated in Figure~\ref{fig:daily}, black curve, which is far smoother and hence more realistic from an epidemiological point of view. 
To go further in the modeling of infection counts, the present work proposes to account for the increased variance observed in real COVID-19 infection counts through a scale parameter $\alpha > 1$ constant in time.  Up to the authors' knowledge,  scaled Poisson distributions have not yet been explored in the epidemiology literature, making this an original contribution.
Altogether, the proposed \emph{weekly scaled Poisson} epidemiological model writes
\begin{align}
\label{eq:modif-epi}
\Z_t \mid \Z_1, \hdots, \Z_{t-1} \sim \alpha \mathcal{P}\left(\frac{ \Phi_t(\Zvect)  \R_t}{\alpha} \right).
\end{align}
where time instants $t$ corresponds to weeks and the serial interval distribution used to compute $\Phi_t(\Zvect)$ is \emph{weekly} discretized.
The weekly discretized COVID-19 serial interval distribution can be obtained by coarsening the daily discretized distribution provided in~\cite{guzzetta2020,riccardo2020} into a weekly distribution, e.g., by performing an integration over one-week windows using the left rectangle method with a one-day integration step.
Under Model~\eqref{eq:modif-epi}, the expected number of infections at week $t$ is $\Phi_t(\Zvect)  \R_t$, unchanged compared to the standard Poisson model of Equation~\eqref{eq:model_epi}, but the variance is $\alpha \Phi_t(\Zvect)  \R_t$. If $\alpha > 1$, then the variance is larger by a factor $\alpha$ than the standard Poisson model variance.

\begin{figure}
\flushleft
\hspace{1.25cm}\begin{subfigure}{0.9\linewidth}
\centering
\includegraphics[width =0.7 \linewidth]{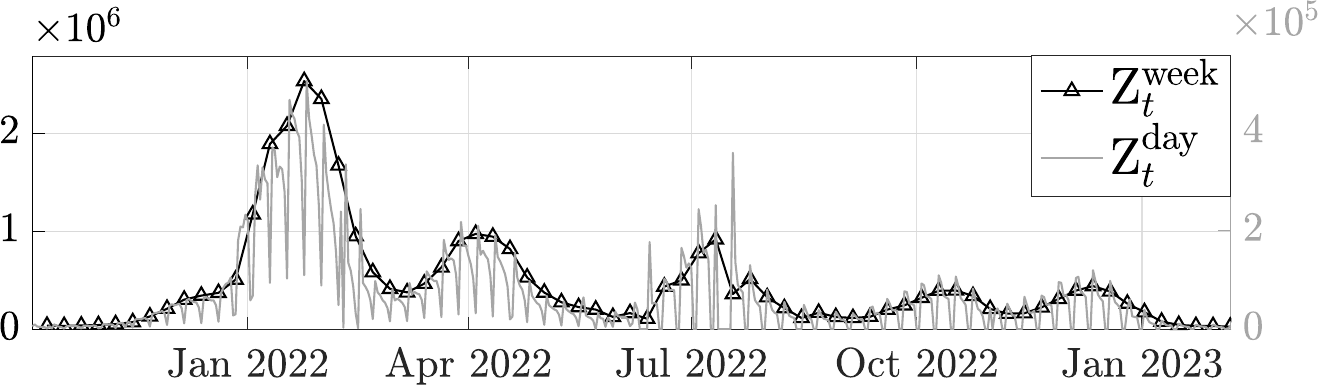}
\end{subfigure}

\vspace{-2mm}
\caption{\label{fig:daily}\textbf{COVID-19 daily vs. weekly infection counts.} French counts from Johns Hopkins University's repository,   covering 71 weeks from October 4, 2021 to February 6, 2023.}
\vspace{-3.5mm}
\end{figure}

\subsection{Data-driven reproduction number estimation strategy} 
\label{ssec:data-driven-Rhat}

Following~\cite{abry2020spatial,pascal2022nonsmooth}, $\R_t$ is assumed piecewise linear in time,  and the COVID-19 \emph{weekly} reproduction number is thus estimated from aggregated infection counts by plugging the weekly discretized propagation model introduced in Equation~\eqref{eq:modif-epi} into the variational estimator of Equation~\eqref{eq:practical-est} leading to
\begin{align}
\label{eq:covid-est}
\widehat{\Rvect}(\Zvect;\param) = \underset{\Rvect \in \mathbb{R}_+^T}{\mathrm{argmin}} \, \, \mathcal{D}_{\alphavect}\left(\Zvect , \Rvect \odot\Phivect(\Zvect)\right) + \param \lVert \textbf{D}_2 \Rvect \rVert_1.
\end{align}
The scale parameter, unknown in practice, is assumed constant in time. 
Remark that $\Phi_t(\Zvect)$ is linear in $\Zvect$, with $\sum_{s \geq 1} \varphi_s = 1$ by normalization of the serial interval distribution, hence if $s > t$, then $\partial_{\Z_t} \Phi_s(\Zvect)=\varphi_{s-t}$ is of order one.
Thus, to ensure that $\alpha$ is reflecting the inflated variance observed in COVID-19 data while satisfying Assumption~\ref{hyp:psi-alpha}, the following data-driven heuristics is proposed $\alpha = 0.1 \times \mathrm{std}(\Zvect)$, where $\mathrm{std}(\Zvect)$ is the empirical standard deviation of $\Zvect$.
The regularization parameter $\lambda$ controlling the level of regularization in Equation~\eqref{eq:covid-est} is selected in a data-driven manner through
\begin{align}
\label{eq:optimal-covid}
\lambda_{\overline{\mathsf{APURE}^{\mathcal{P}}_{\mc}}^N}  \in \underset{\param \in \mathbb{R_+}}{\mathrm{Argmin}} \, \, \overline{\mathsf{APURE}^{\mathcal{P}}_{\mc}}^N(\Zvect;\param)
\end{align}
leveraging the robustified unbiased prediction risk estimate of Equation~\eqref{eq:fdmc-apure}, whose accuracy and robustness have been demonstrated on synthetic data in Section~\ref{sec:numerical}.

\subsection{Experimental setup}

During the entire course of the COVID-19 pandemic, daily infection counts reported by National Health Agencies of 200+ countries worldwide were collected by Johns Hopkins University and made publicly available in an online repository.\footref{fn:jhu}
To demonstrate the universality of the proposed data-driven reproduction number estimation procedure,  four countries, belonging to different continents, and two time periods, corresponding to late and early epidemic stages, are considered.
Daily reported counts in France and India, from October 5, 2021 to February 6, 2023, and in Canada and Argentina, from December 22, 2020 to April 25, 2022,  are downloaded from JHU repository.\footnote{\label{fn:jhu}\url{https://coronavirus.jhu.edu/}}
Each sequence of daily counts is then aggregated at the scale of the week, yielding a vector of observed counts $\Zvect$ of length $T = 70$~weeks.
Finally, the associated global infectiousness $\Phivect(\Zvect)$ is computed using $\Phi_t(\Zvect) = \sum_{s \geq 1} \varphi_s \Z_{t-s}$, where $\varphi$ is the weekly discretized serial interval distribution introduced in Section~\ref{ssec:weekly-model}.
Top plots of Figure~\ref{fig:covid} show weekly infection counts and associated global infectiousness in France~\ref{sfig:France} and India~\ref{sfig:India}, from 2021 to 2023,  and Canada~\ref{sfig:Canada} and Argentina~\ref{sfig:Argentina}, from 2020 to 2022.

The minimization in Equation~\eqref{eq:covid-est} is performed using the Chambolle-Pock algorithm derived in~\cite{abry2020spatial,pascal2022nonsmooth},  as in Section~\ref{ssec:est-strat}.
Based on Section~\ref{sec:numerical}, the robustified unbiased prediction risk estimate is computed by averaging over $N=10$ realizations of the Monte Carlo vector, which yields a robust approximation of the prediction error, and accurate estimates when used as an oracle for hyperparameter selection.
The regularization parameter is chosen by solving~\eqref{eq:optimal-covid} through an exhaustive search over a logarithmic grid of hyperparameters $\lambda$ ranging from $10^{-2} \times \mathrm{std}(\Zvect)$ to $10^4 \times \mathrm{std}(\Zvect)$.

\begin{figure}[t!]
\centering
\begin{subfigure}{0.49\linewidth}
\includegraphics[width = \linewidth]{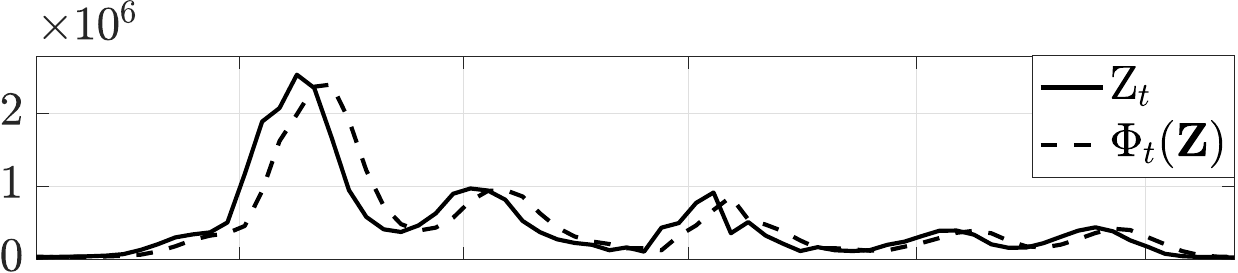}

\vspace{1mm}
\includegraphics[width = \linewidth]{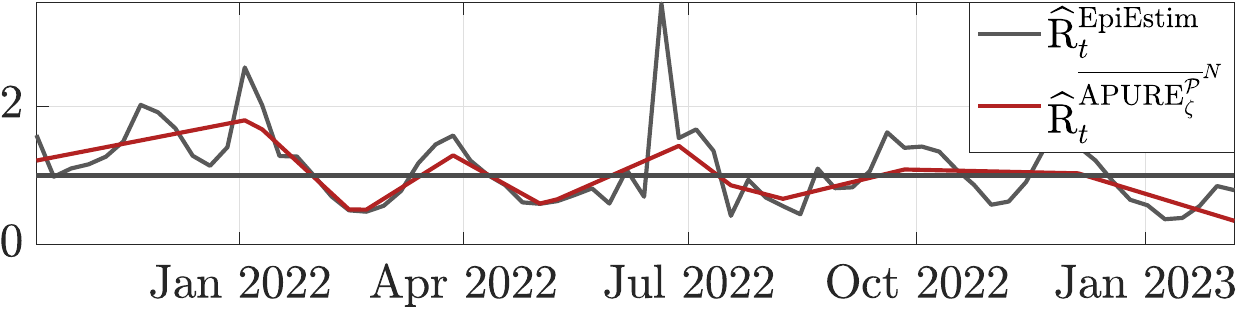}

\vspace{1mm}
\includegraphics[trim =  2.7cm 0mm 2.1cm 0mm, clip,width = \linewidth]{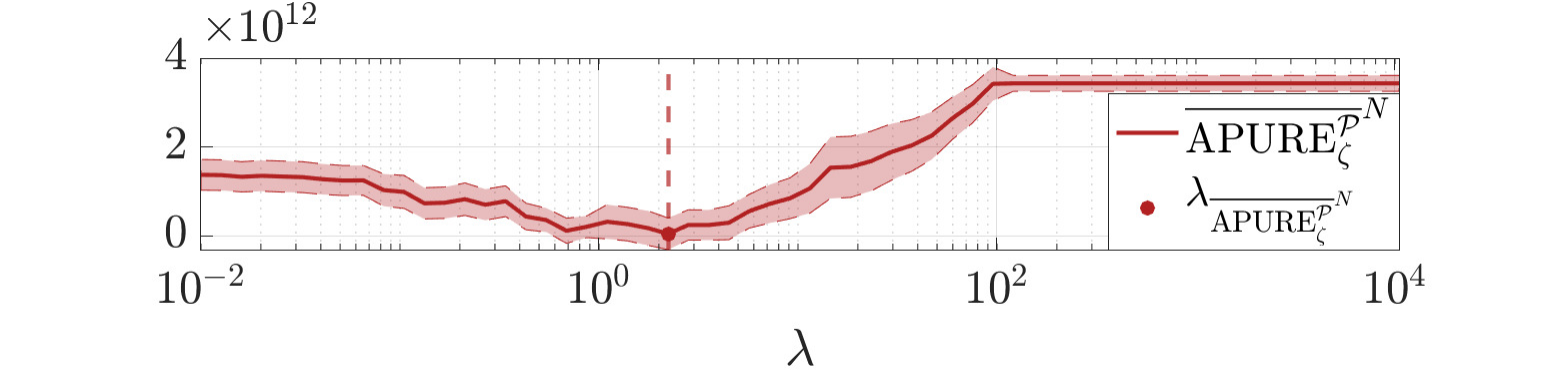}

\vspace{-2mm}
\subcaption{\label{sfig:France}France} 
\end{subfigure}
\begin{subfigure}{0.49\linewidth}
\includegraphics[width = \linewidth]{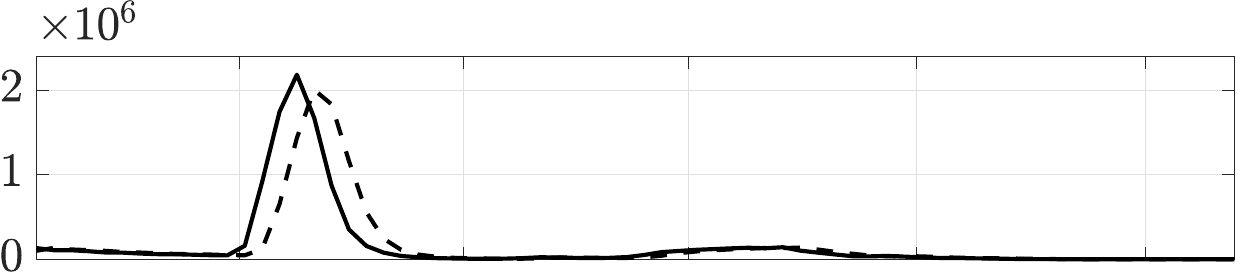}

\vspace{1mm}
\includegraphics[width = \linewidth]{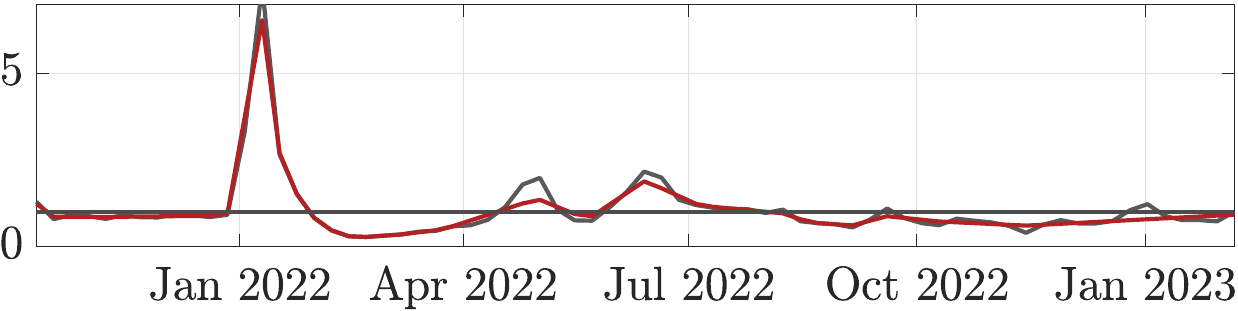}

\vspace{1mm}
\includegraphics[trim =  2.7cm 0mm 2.1cm 0mm, clip,width = \linewidth]{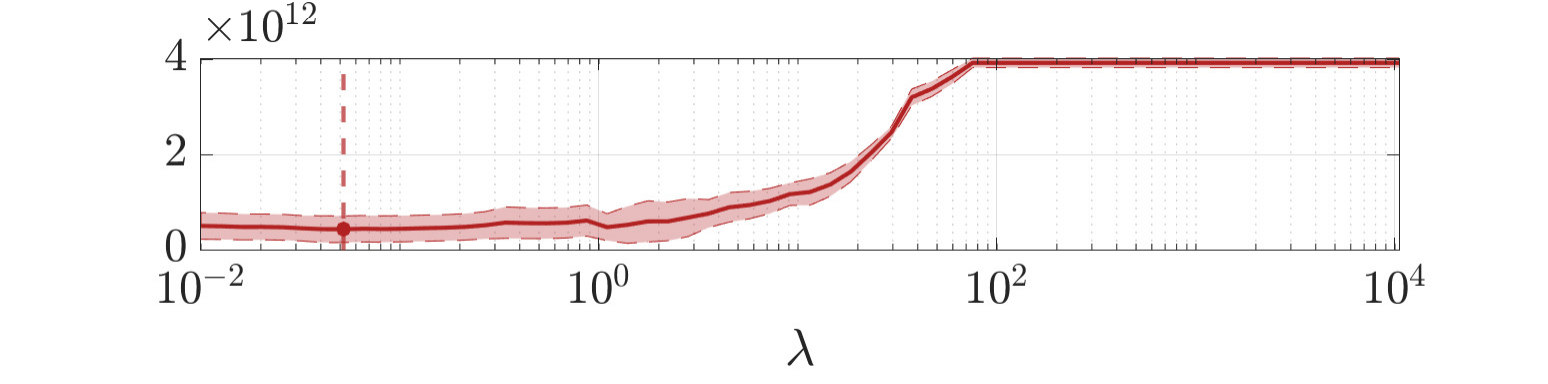}

\vspace{-2mm}
\subcaption{\label{sfig:India}India}
\end{subfigure}

\vspace{-1mm}
\begin{subfigure}{0.49\linewidth}
\includegraphics[width = \linewidth]{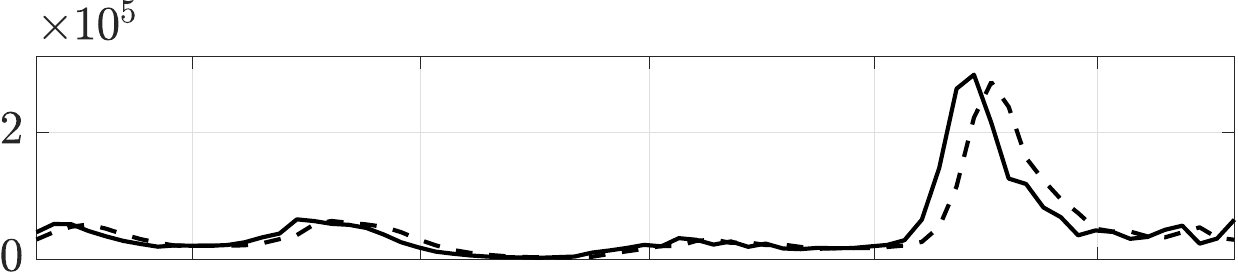}

\vspace{1mm}
\includegraphics[width = \linewidth]{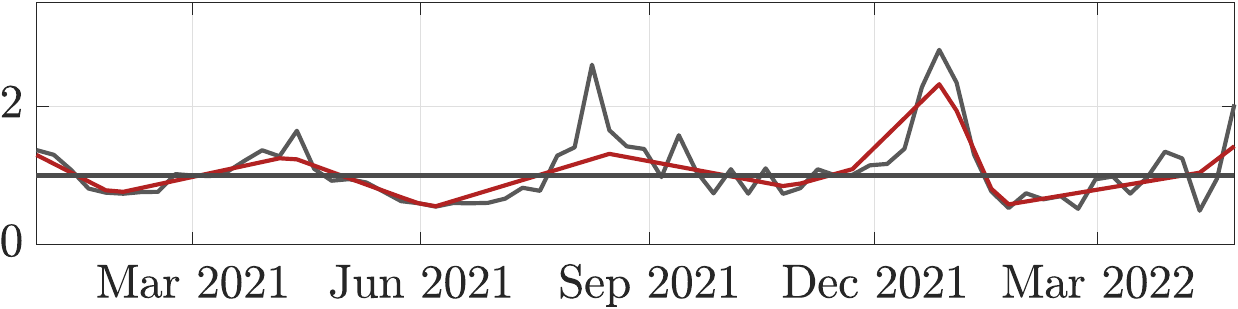}

\vspace{1mm}
\includegraphics[trim =  2.7cm 0mm 2.1cm 0mm, clip,width = \linewidth]{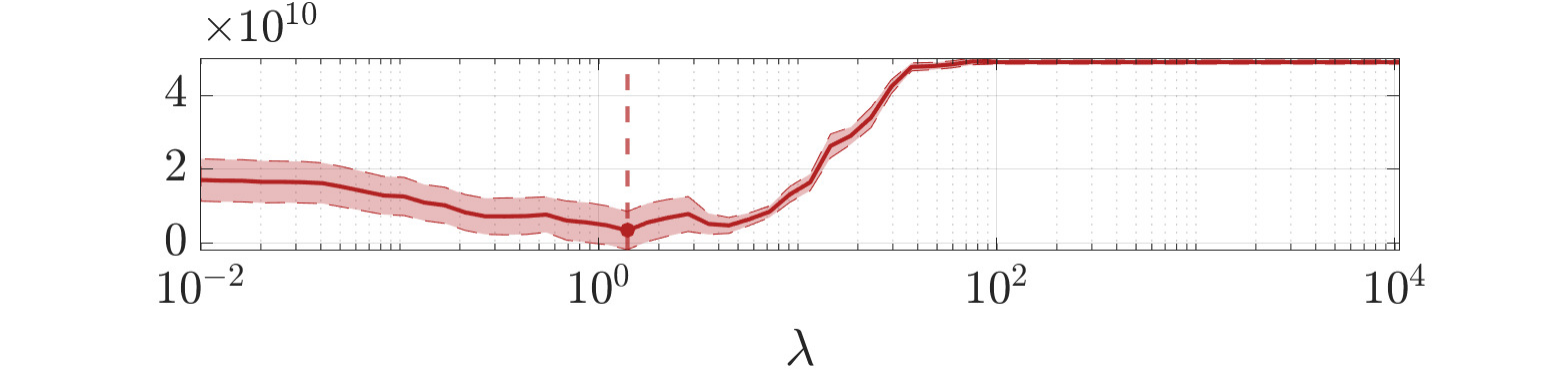}

\vspace{-2mm}
\subcaption{\label{sfig:Canada}Canada}
\end{subfigure}
\begin{subfigure}{0.49\linewidth}
\includegraphics[width = \linewidth]{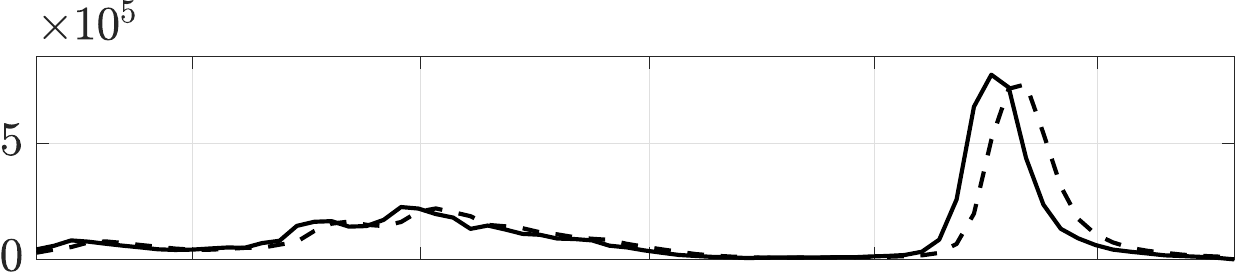}

\vspace{1mm}
\includegraphics[width = \linewidth]{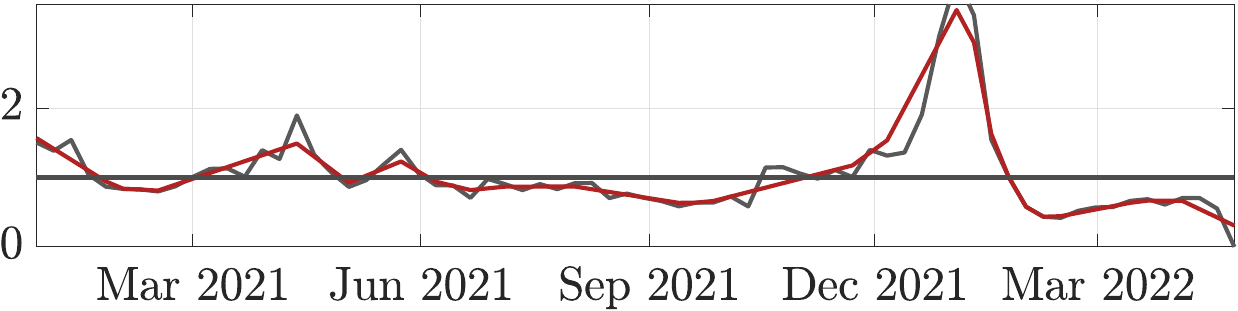}

\vspace{1mm}
\includegraphics[trim =  2.7cm 0mm 2.1cm 0mm, clip,width = \linewidth]{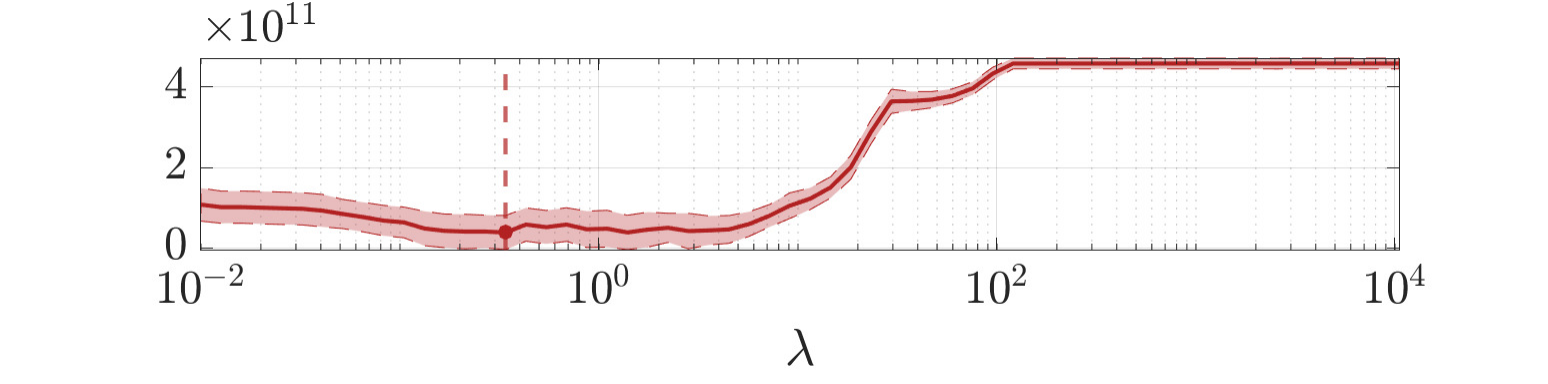}

\vspace{-2mm}
\subcaption{\label{sfig:Argentina}Argentina}
\end{subfigure}

\vspace{-3mm}
\caption{\label{fig:covid}\textbf{Data-driven estimation of the reproduction number of COVID-19 from weekly infection counts.} 
\underline{Top rows:} weekly aggregated infection counts. 
\underline{Middle rows:} estimated effective reproduction number with state-of-the-art EpiEstim~\cite{cori2013new,thompson2019improved} in dark gray \emph{versus} data-driven hyperparameter selection based on the robustified $\overline{\mathsf{APURE}^{\mathcal{P}}_{\mc}}^N$ in red. 
\underline{Bottom rows:} robustified autoregressive unbiased prediction risk estimate averaged over $N=10$ realizations of the Monte Carlo vector with $95\%$ Gaussian confidence interval for a logarithmically spaced range of regularization parameters and minimizing regularization parameter obtained through grid search.}
\vspace{-3.5mm}
\end{figure}

\subsection{Discussion}

For each country and pandemic phase configuration, the robustified unbiased prediction risk estimate as a function of the regularization parameter $\lambda$, accompanied by the associated 95$\%$ Gaussian confidence interval computed from the $N$ realizations of the Monte Carlo vector,  are displayed  in the bottom plots of Figure~\ref{fig:covid}.
The estimated optimal hyperparameters are indicated by the vertical dashed lines. 
For all four configurations, from Figure~\ref{sfig:France} to~\ref{sfig:Argentina}, the prediction risk estimate has a clearly identifiable minimizer.
The resulting reproduction number estimates are displayed as red curves in the middle plots of Figure~\ref{fig:covid}, together with the \emph{weekly} EpiEstim estimator, corresponding to Equation~\eqref{eq:epiestim} with $\delta = 1$ since data are already aggregated at the scale of the week, represented as dark gray curves. 
It is worth noting that for weekly aggregated counts, the state-of-the-art reproduction number estimator proposed in~\cite{cori2013new,thompson2019improved} almost exactly boils down to the Maximum Likelihood estimate~\eqref{eq:ML-est} since, first $\delta = 1$ induces no smoothing of data and, second  $a=1$ and $1/b = 1/5$ are very small compared to $\Z_t$ and $\Phi_t(\Zvect)$ which are of order $10^5$ to $10^6$ as shown in the top row plots in Figure~\ref{fig:covid}. 
Though, as far as the authors know, no other refined estimator of \emph{weekly} reproduction numbers has been proposed up to now.
The regularized estimates with data-driven hyperparameter selection, in red, vary more regularly than the EpiEstim estimates, in dark gray,  which exhibit unrealistic abrupt changes caused by administrative noise.
Hence, the regularized estimates  account more faithfully for the pandemic spread.
This is notably the case for France, Figure~\ref{sfig:France}, and Canada, Figure~\ref{sfig:Canada}, where the EpiEstim estimate presents rapid fluctuations leading to significant overestimation of the reproduction number,  e.g.,  in July 2022 in France and August 2021 in Canada.
This is all the more important to get accurate estimate of $\R_t$ as the sanitary measures have dramatic social and economical impact that decision makers need to balance.
In such context, both false alarms and missed events may have highly detrimental consequences.
Figure~\ref{sfig:India} shows that the proposed data-driven reproduction number estimator is able to capture both rapid bursts, corresponding to severe pandemic waves, as experimented by India in January 2022, and smaller waves, e.g., in June 2022, with similar accuracy.
Together with the quantitative assessment performed on synthetic data in Section~\ref{sec:numerical}, this qualitative assessment of the proposed data-driven reproduction number estimation strategy on real COVID-19 infection counts of different countries in various pandemic stages, demonstrates its ability to be used to monitor closely a viral epidemics, even in a context of degraded reporting suffering from outlier samples and inflated variance.
Finally, regarding computational complexity, estimating the weekly reproduction number during a time-period of $70$ weeks in a given country, that is obtaining one result among the four presented in Figure~\ref{fig:covid}, requires less than $50$ minutes on a laptop equipped with a 2 GHz Intel Core i5 processor. The proposed overall pipeline is thus easily usable for practical epidemic monitoring.

\section{Conclusion and perspectives}
\label{sec:conclusion}
A time-varying autoregressive model generalizing state-of-the-art viral epidemics models has been considered.
The estimation of the parameters of the model, namely the  time-varying \emph{reproduction coefficient},  has been framed as a nonstandard, highly nonstationary, inverse problem.
The proposed rigorous mathematical formulation enables to leverage the efficient and versatile variational framework to design accurate reproduction coefficient estimators  which are robust to high noise levels in the observations.
The major contribution consists in the design of asymptotically unbiased risk estimates, suited to the generalized autoregressive Poisson model, which are then plugged into an oracle strategy for fully data-driven fine-tuning of the  hyperparameters of the variational estimator, removing the major obstacle to its practical use.
The resulting data-driven time-varying estimation strategy has been assessed through intensive Monte Carlo simulations on synthetic data.
Taking advantage of the proposed extended mathematical framework, the standard epidemiological model for viral epidemics has been enriched to account for the low quality of COVID-19 data, leading to a novel \emph{weekly scaled Poisson} model. 
Finally, the data-driven time-varying estimation procedure is shown to yield very consistent estimation of the COVID-19 reproduction number in various countries and pandemic stages despite the low quality of reported data, demonstrating its practical applicability for epidemic monitoring in a crisis context.
The data-driven nature of the proposed epidemiological indicator estimation strategy and its low computational cost constitute major assets for its dissemination beyond signal processing as its use requires no expert knowledge.

Future work will consist in, first, implementing and assessing the unbiased risk estimates for the other noise models, namely additive Gaussian and multiplicative Gamma noises extending the generalized Stein estimators~\cite{eldar2008generalized}, in order to further enrich the proposed framework and to be able to apply the developed methodology to other models,  in epidemiology or beyond.
Second,  as in epidemiology the memory functions are often parametric, with parameters encapsulating the pathogen transmission characteristics, the proposed framework will be leveraged to perform simultaneously the extraction of the memory functions parameters and the fine-tuning of the estimator hyperparameters.

For the sake of reproducibility and dissemination, a Matlab toolbox is publicly available on the GitHub of the corresponding author at \url{https://github.com/bpascal-fr/APURE-Estim-Epi}.

%\section{Acknowledgements}

%% The Appendices part is started with the command \appendix;
%% appendix sections are then done as normal sections
\appendix
%\section{Example Appendix Section}
%\label{app1}

%\clearpage
%\newpage

\section{Autoregressive Poisson Stein-like lemma}
\label{ssec:proof-lemma}

\begin{proof}[Proof of generalized time-varying autoregressive Poisson Stein-like lemma~\ref{lem:Poisson-Stein}]
Let $\Yvect \in \mathbb{R}_+^T$ random observations under the generalized time-varying autoregressive model~\eqref{eq:model} with ground truth reproduction number $\overline{\Xvect} \in \mathbb{R}_+^T$ and memory functions $\lbrace \Psi_t, \, t = 1, \hdots, T \rbrace$ satisfying Assumption~\ref{hyp:psi-alpha}, following a scaled Poisson distribution~\eqref{eq:poiss_model} of time-varying modulus $\alphavect\in (\mathbb{R}_+^*)^T$, and
a function $\Theta : \mathbb{R}_+^T \rightarrow\mathbb{R}$ satisfying Assumption~\ref{hyp:Poisson}.
First,  by Assumption~\ref{hyp:Poisson},  the expectation in the left-hand side of Equation~\eqref{eq:poisson-lemma} is well defined.
Hence, using the discrete probabilistic density function of the scaled Poisson distribution leads to
\begin{align}
\mathbb{E}_{\Yvect}  \left[ \Theta(\Yvect) \overline{\X}_t \Psi_t(\Yvect) \right] \label{eq:full_expect} 
=\sum_{\kk_1 = 0}^\infty \!  \hdots \!\! \sum_{\kk_T = 0}^\infty \!\! \Theta(\yvect) \overline{\X}_t \Psi_t(\yvect)  \prod_{s = 1}^T \frac{(\overline{\X}_s \Psi_s(\yvect))^{ \kk_s}}{\kk_s!}\mathrm{e}^{-\overline{\X}_s \Psi_s(\yvect)},
\end{align}
where  $\forall t, \, \y_t = \alpha_t \kk_t$, $\kk_t \in \mathbb{N}$.
Since $\Psi_s(\yvect)$ does not depend on $\y_s$ for $s \geq u$,  considering only the sums over $\kk_t, \hdots, \kk_T$, the above expression factorizes as
\begin{align}
%\begin{split}
\sum_{\kk_{t} = 0}^\infty \!  \hdots \!\! \sum_{\kk_T = 0}^\infty \!\!  \Theta(\yvect) \overline{\X}_t \Psi_t(\yvect)  \prod_{s = 1}^T \frac{(\overline{\X}_s \Psi_s(\yvect))^{ \kk_s}}{\kk_s!}\mathrm{e}^{-\overline{\X}_s \Psi_s(\yvect)} 
= \,   \prod_{s = 1}^{t-1} \frac{(\overline{\X}_s \Psi_s(\yvect))^{ \kk_s}}{\kk_s!}\mathrm{e}^{-\overline{\X}_s \Psi_s(\yvect)}  \label{eq:sum_kt}
 \times  \sum_{k_t = 0}^\infty \mathsf{G}(\yvect)  \overline{\X}_t \Psi_t(\yvect)  \frac{(\overline{\X}_t \Psi_t(\yvect))^{ \kk_t}}{\kk_t!} \mathrm{e}^{-\overline{\X}_t \Psi_t(\yvect)}  
%\end{split}
 \end{align}
\begin{align}
\label{eq:def_Fy}
\text{where} \quad \mathsf{G}(\yvect) = \sum_{\kk_{t+1} = 0}^\infty \!  \hdots \!\! \sum_{\kk_T = 0}^\infty \!\! \Theta(\yvect) \prod_{u = t+1}^T \frac{(\overline{\X}_s \Psi_s(\yvect))^{ \kk_s}}{\kk_s!}\mathrm{e}^{-\overline{\X}_s \Psi_s(\yvect)} .
\end{align}
$\mathsf{G}(\yvect)$ is obtained by marginalizing over the variables $\y_{t+1},\hdots, \y_T$,  and thus it depends only on $\y_1, \hdots, \y_t$.
Then, the summation over $\kk_t$ appearing in~\eqref{eq:sum_kt} writes
\begin{align}
\label{eq:sum_ktp1}
 \sum_{\kk_t = 0}^\infty  \mathsf{G}(\yvect) \overline{\X}_t \Psi_t(\yvect)\frac{( \overline{\X}_t \Psi_t(\yvect))^{\kk_t}}{\kk_t!} \mathrm{e}^{-\overline{\X}_t \Psi_t(\yvect)}  
=   \sum_{\kk_t = 0}^\infty  \mathsf{G}(\yvect) \times (\kk_t + 1) \times  \frac{( \overline{\X}_t \Psi_t(\yvect))^{{\kk_t+1}}}{(\kk_t+1)!} \mathrm{e}^{-\overline{\X}_t \Psi_t(\yvect)}  
%= &  \sum_{\kk_t = 0}^\infty  \mathsf{G}^{-t}(\yvect)  \widetilde{\kk}_t \frac{( \overline{\X}_t \Psi_t(\widetilde{\yvect}))^{{\widetilde{\kk}_t}}}{\widetilde{\kk}_t!}  \mathrm{e}^{-\overline{\X}_t \Psi_t(\widetilde{\yvect})}
\end{align}
where, by Definition~\ref{def:model},  $\Psi_t(\yvect) = \Psi_t(\alpha_1 \kk_1, \hdots, \alpha_{t-1}\kk_{t-1})$\footnote{By convention $\Psi_1 = \Y_0$, where $\Y_0$ is a deterministic initialization.} is not depending on $\kk_t$, and $\mathsf{G}(\yvect) = \mathsf{G}(\alpha_1\kk_1, \hdots, \alpha_T \kk_T)$.
Renaming the summation variable in the right-hand side of Equation~\eqref{eq:sum_ktp1} so that $\kk_t + 1 \rightarrow \kk_t$, it follows
\begin{align}
\label{eq:sum_ktp2}
 \sum_{\kk_t = 0}^\infty  &\mathsf{G}(\yvect) \overline{\X}_t \Psi_t(\yvect)\frac{( \overline{\X}_t \Psi_t(\yvect))^{\kk_t}}{\kk_t!} \mathrm{e}^{-\overline{\X}_t \Psi_t(\yvect)}  
=   \sum_{\kk_t = 0}^\infty  \mathsf{G}^{-t}(\yvect)  \kk_t \frac{( \overline{\X}_t \Psi_t(\yvect))^{{\kk_t}}}{\kk_t!}  \mathrm{e}^{-\overline{\X}_t \Psi_t(\yvect)}
\end{align}
where $\Psi_t(\yvect)$ is unchanged as it depends only on the \emph{fixed} variables $\y_1, \hdots, \y_{t-1}$, and, by definition, $\mathsf{G}^{-t}(\yvect) = \mathsf{G}(\y_1, \hdots, \y_t - \alpha_t, \hdots, \y_t)$. 

\begin{remark}
The above discrete counterpart of integration by part on variable $\kk_t$ performed in Equations~\eqref{eq:sum_ktp1} and~\eqref{eq:sum_ktp2} can be seen as an application of the standard Poisson Stein's lemma counterpart~\cite{luisier2007new,le2014unbiased} on the function $\mathsf{G}$. Though, this reformulation is not enough to design Poisson Unbiased Risk Estimates for the generalized time-varying autoregressive model~\eqref{eq:model} as, contrary to the case of standard Poisson Unbiased Risk Estimate, neither the function $\mathsf{G}$ nor its translate versions $\mathsf{G}^{-t}$ can be reformulated easily as expectations of tractable quantities.
\end{remark}

Further computations are thus required.  By Assumption~\ref{hyp:psi-alpha}:
\begin{align}
\! \! \! \! \forall u \geq t+1, \quad \Psi_s(\widetilde{\yvect}) \simeq \Psi_s(\yvect) - \partial_{\y_t} \Psi_s(\yvect) \alpha_t \simeq  \Psi_s(\yvect),
\end{align}
which, when injected into the expression of $\mathsf{G}^{-t}(\yvect)$, obtained by replacing $\yvect = (\y_1, \hdots, \y_T)$ by $\widetilde{\yvect} = (\y_1, \hdots, \y_t - \alpha_t, \hdots, \y_T)$ in Equation~\eqref{eq:def_Fy},  yields $\mathsf{G}^{-t}(\yvect)\simeq$
\begin{align}
 \label{eq:sum_yt}
\! \! \!\sum_{\kk_{t+1} = 0}^\infty \!  \hdots \!\! \sum_{\kk_T = 0}^\infty \!\! \Theta^{-t}(\yvect) \prod_{u = t+1}^T \! \! \frac{(\overline{\X}_s \Psi_s(\yvect))^{ \kk_s}}{\kk_s!} \mathrm{e}^{-\overline{\X}_s \Psi_s(\yvect)}
\end{align}
where the change of variable $\yvect \rightarrow \widetilde{\yvect}$ only affects the term in $\Theta$ but not the Poisson densities.
Finally, injecting~\eqref{eq:sum_yt} into~\eqref{eq:full_expect}, yields 
\begin{align}
\mathbb{E}_{\Yvect}  \left[ \Theta(\Yvect) \overline{\X}_t \Psi_t(\Yvect) \right] \simeq \sum_{\kk_1 = 0}^\infty \!  \hdots \!\! \sum_{\kk_T = 0}^\infty \!\!  \Theta^{-t}(\yvect)  \kk_t \prod_{s = 1}^T \frac{(\overline{\X}_s \Psi_s(\yvect))^{ \kk_s}}{\kk_s!} \mathrm{e}^{-\overline{\X}_s \Psi_s(\yvect)}
\end{align}
with $\forall t, \, \y_t = \alpha \kk_t$, in which one recognizes $ \mathbb{E}_{\Yvect}  \left[ \Theta^{-t}(\Yvect) \Y_t \right]$ by definition of the expectation on $\Yvect$.
\end{proof}

\section{Autoregressive Poisson Unbiased Risk Estimate}
\label{ssec:apure-proof}
\begin{proof}[Proof of Theorem~\ref{thm:pure-poisson}]

Let $\mathcal{P}$ be the prediction risk defined in Equation~\eqref{eq:pred_risk}. The derivation of an unbiased estimate of $\mathcal{P}$ relies on the expansion of the prediction error:
\begin{align}
%\begin{split}
\left\lVert  \widehat{\Xvect}(\Yvect ; \paramvect) \odot \Psivect(\Yvect)  - \overline{\Xvect} \odot \Psivect(\Yvect) \right\rVert^2_2 =\left\lVert   \widehat{\Xvect}(\Yvect ; \paramvect)  \odot \Psivect(\Yvect) \right\rVert^2_2   - 2 \left\langle\widehat{\Xvect}(\Yvect ; \paramvect)  \odot \Psivect(\Yvect) ,   \overline{\Xvect} \odot \Psivect(\Yvect)\right\rangle + 
 \left\lVert \overline{\Xvect} \odot \Psivect(\Yvect)  \right\rVert^2_2. 
 \label{eq:expand_first}
% \end{split}
\end{align}
The first term in the right-hand side of~\eqref{eq:expand_first} only depends on observations and hyperparameters, and hence appears as is in $\mathsf{APURE}^{\mathcal{P}}$ in Equation~\eqref{eq:apure}.
The second and third terms depend on the inaccessible ground truth $\overline{\Xvect}$, and hence need to be reformulated, taking care of not introducing any bias.
Considering the second term, 
\begin{align*}
\mathbb{E}_{\Yvect} \left[\left\langle\widehat{\Xvect}(\Yvect;\paramvect) \odot \Psivect(\Yvect) ,   \overline{\Xvect} \odot \Psivect(\Yvect)\right\rangle \right] 
=  \sum_{t = 1}^T \mathbb{E}_{\Yvect} \left[\widehat{\X}_t(\Yvect;\paramvect)  \Psi_t(\Yvect)  \overline{\X}_t  \Psi_t(\Yvect) \right].
\end{align*}
By hypothesis, the memory functions satisfy Assumption~\ref{hyp:psi-alpha}. Further, for any $\paramvect \in \paramset$, and any $t \in \lbrace 1, \hdots, T \rbrace$, the function $\Yvect \mapsto \widehat{\X}_t(\Yvect;\paramvect)\Psi_t(\Yvect)$ satisfies Assumption~\ref{hyp:Poisson}, hence the autoregressive Poisson lemma~\ref{lem:Poisson-Stein} applies and yields
\begin{align}
\begin{split}
 \mathbb{E}_{\Yvect} \left[\widehat{\X}_t(\Yvect;\paramvect)  \Psi_t(\Yvect)  \overline{\X}_t  \Psi_t(\Yvect) \right] 
  \underset{\alphavect \rightarrow \boldsymbol{0}}{=}  \mathbb{E}_{\Yvect} \left[\left(\widehat{\X}_t(\Yvect;\paramvect)  \Psi_t(\Yvect) \right)^{-t} \Y_t \right]
 \end{split}
\end{align}
where $\left(\widehat{\X}_t(\Yvect;\paramvect) \Psi_t(\Yvect)\right)^{-t} = \widehat{\X}_t^{-t}(\Yvect;\paramvect) \Psi_t^{-t}(\Yvect)$.
Since $\Psi_t$ does not depend on $\Y_t$,  $\Psi_t^{-t}(\Yvect)=\Psi_t(\Yvect)$, leading to
\begin{align}
\begin{split}
\label{eq:second-term-pred}
 \mathbb{E}_{\Yvect}  \left[\widehat{\X}_t(\Yvect;\paramvect)  \Psi_t(\Yvect)  \overline{\X}_t  \Psi_t(\Yvect) \right] 
 \underset{\alphavect \rightarrow \boldsymbol{0}}{=}  \mathbb{E}_{\Yvect} \left[\widehat{\X}_t^{-t}(\Yvect;\paramvect)  \Psi_t(\Yvect)  \Y_t \right]
 \end{split}
\end{align}
in which the reader recognizes the expression of the second term in the definition of $\mathsf{APURE}^{\mathcal{P}}$ in Equation~\eqref{eq:apure}.
As for the third term of Equation~\eqref{eq:expand_first}, it writes
\begin{align}
\mathbb{E}_{\Yvect} \left[  \left\lVert \overline{\Xvect} \odot \Psivect(\Yvect)  \right\rVert^2_2 \right] = \sum_{t = 1}^T  \mathbb{E}_{\Yvect} \left[ \lvert \overline{\X}_t \Psi_t(\Yvect)\rvert^2 \right].
\end{align}
By hypothesis, the memory functions satisfy Assumption~\ref{hyp:psi-alpha}. Moreover for all $t \in \lbrace 1, \hdots, T \rbrace$,  the function $\Yvect \mapsto \overline{\X}_t \Psi_t(\Yvect)$ satisfies Assumption~\ref{hyp:Poisson}, hence the autoregressive Poisson lemma~\ref{lem:Poisson-Stein} applies and yields
\begin{align}
 \mathbb{E}_{\Yvect} \left[ ( \overline{\X}_t \Psi_t(\Yvect))^2 \right] \underset{\alphavect \rightarrow \boldsymbol{0}}{=}  \mathbb{E}_{\Yvect} \left[  \left(\overline{\X}_t \Psi_t(\Yvect)\right)^{-t}   \Y_t \right]
\end{align}
where $\left(\overline{\X}_t \Psi_t(\Yvect)\right)^{-t} = \overline{\X}_t \Psi_t^{-t}(\Yvect)$.
But, since $\Psi_t$ does not depend on $\Y_t$,  one has $\overline{\X}_t \Psi_t^{-t}(\Yvect)= \overline{\X}_t \Psi_t(\Yvect)$, and hence
 \begin{align*}
\mathbb{E}_{\Yvect} \left[ ( \overline{\X}_t \Psi_t(\Yvect))^2 \right]  \underset{\alphavect \rightarrow \boldsymbol{0}}{=}  \mathbb{E}_{\Yvect} \left[   \overline{\X}_t \Psi_t(\Yvect) \Y_t  \right] =  \mathbb{E}_{\Yvect} \left[  \Y_t   \overline{\X}_t \Psi_t(\Yvect) \right].
\end{align*}
Then, remarking that the function $\Yvect \mapsto \Y_t$ satisfies Assumption~\ref{hyp:Poisson},  the autoregressive Poisson Stein's lemma~\ref{lem:Poisson-Stein} applies once again, and it follows that
\begin{align}
\begin{split}
\label{eq:third-term-pred}
\mathbb{E}_{\Yvect} \left[ ( \overline{\X}_t \Psi_t(\Yvect))^2 \right] \underset{\alphavect\rightarrow \boldsymbol{0}}{=}  \mathbb{E}_{\Yvect} \left[   \left( \Y_t - \alpha_t \right)  \Y_t\right].
\end{split}
\end{align}
Equations \eqref{eq:second-term-pred} and~\eqref{eq:third-term-pred}, combined with the expansion provided in Equation~\eqref{eq:expand_first} demonstrates the asymptotic unbiasedness of $\mathsf{APURE}^{\mathcal{P}}$ stated in Equation~\eqref{eq:apure-thm1}.

Assuming that $\forall t\in \lbrace 1, \hdots, T\rbrace, \,  \forall \Yvect \in \mathbb{R}^T, \, \Psi_t(\Yvect) \neq 0$,  the estimation risk $\mathcal{E}$ of Equation~\eqref{eq:est_risk} is well-defined and the estimation error can be expanded as:
\begin{align}
%\begin{split}
\left\lVert  \widehat{\Xvect}(\Yvect ; \paramvect) - \overline{\Xvect} \right\rVert^2_2 =\left\lVert   \widehat{\Xvect}(\Yvect ; \paramvect)   \right\rVert^2_2 
- 2 \left\langle\widehat{\Xvect}(\Yvect ; \paramvect) ,   \overline{\Xvect} \right\rangle + 
 \left\lVert \overline{\Xvect}  \right\rVert^2_2. 
 \label{eq:expand_est}
% \end{split}
\end{align}
The first term of the expansion, which is fully data-dependent, is kept as is in the definition of $\mathsf{APURE}^{\mathcal{E}}$ in Equation~\eqref{eq:est_risk}. The second and  third terms depend on the ground truth and have to be carefully reformulated, while avoiding to introduce any bias.
The second term writes
\begin{align*}
\mathbb{E}_{\Yvect} &\left[\left\langle\widehat{\Xvect}(\Yvect;\paramvect) ,   \overline{\Xvect} \right\rangle \right] =  \sum_{t = 1}^T \mathbb{E}_{\Yvect} \left[\widehat{\X}_t(\Yvect;\paramvect)  \overline{\X}_t  \right].
\end{align*}
By hypothesis, the memory functions satisfy Assumption~\ref{hyp:psi-alpha}. Further,  for any $\paramvect \in \paramset$, and any $t \in \lbrace 1, \hdots, T \rbrace$, the function $\Yvect \mapsto \widehat{\X}_t(\Yvect;\paramvect)/\Psi_t(\Yvect)$ is well-defined and satisfies Assumption~\ref{hyp:Poisson}, hence the autoregressive Poisson lemma~\ref{lem:Poisson-Stein} applies and yields
%\begin{align}
%\begin{split}
%\label{eq:second-term-est}
% \mathbb{E}_{\Yvect} \left[\widehat{\X}_t(\Yvect;\paramvect)  \overline{\X}_t \right]  =\, \,  \, &\mathbb{E}_{\Yvect} \left[\frac{ \widehat{\X}_t(\Yvect;\paramvect)}{\Psi_t(\Yvect)} \overline{\X}_t \Psi_t(\Yvect) \right]\\
%  \underset{\alphavect \rightarrow \boldsymbol{0}}{=}  &\mathbb{E}_{\Yvect} \left[\left(\frac{\widehat{\X}_t(\Yvect;\paramvect)}{  \Psi_t(\Yvect)} \right)^{-t} \Y_t \right]\\
%   \underset{\alphavect \rightarrow \boldsymbol{0}}{=}  &\mathbb{E}_{\Yvect} \left[\frac{\widehat{\X}_t^{-t}(\Yvect;\paramvect)}{  \Psi_t(\Yvect)}  \Y_t \right]
% \end{split}
%\end{align}
\begin{align}
\label{eq:second-term-est}
 \mathbb{E}_{\Yvect} \left[\widehat{\X}_t(\Yvect;\paramvect)  \overline{\X}_t \right]  =\, \,  \, \mathbb{E}_{\Yvect} \left[\frac{ \widehat{\X}_t(\Yvect;\paramvect)}{\Psi_t(\Yvect)} \overline{\X}_t \Psi_t(\Yvect) \right]
  \underset{\alphavect \rightarrow \boldsymbol{0}}{=}  \mathbb{E}_{\Yvect} \left[\left(\frac{\widehat{\X}_t(\Yvect;\paramvect)}{  \Psi_t(\Yvect)} \right)^{-t} \Y_t \right]
   \underset{\alphavect \rightarrow \boldsymbol{0}}{=}  \mathbb{E}_{\Yvect} \left[\frac{\widehat{\X}_t^{-t}(\Yvect;\paramvect)}{  \Psi_t(\Yvect)}  \Y_t \right]
\end{align}
since $\Psi_t(\Yvect)$ does not depend on $\Y_t$. The third term of Equation~\eqref{eq:expand_est} writes
\begin{align}
\mathbb{E}_{\Yvect} \left[  \left\lVert \overline{\Xvect}  \right\rVert^2_2 \right] = \sum_{t = 1}^T  \mathbb{E}_{\Yvect} \left[   \overline{\X}_t^2 \right] = \sum_{t = 1}^T  \mathbb{E}_{\Yvect} \left[  \frac{\overline{\X}_t}{\Psi_t(\Yvect)} \overline{\X}_t \Psi_t(\Yvect)  \right].
\end{align}
By hypothesis, the memory functions satisfy Assumption~\ref{hyp:psi-alpha}. Moreover, for any $\paramvect \in \paramset$, and any $t \in \lbrace 1, \hdots, T \rbrace$, the function $\Yvect \mapsto \widehat{\X}_t(\Yvect;\paramvect)/\Psi_t(\Yvect)$ is well-defined and satisfies Assumption~\ref{hyp:Poisson}, hence the autoregressive Poisson lemma~\ref{lem:Poisson-Stein} applies,  using that $\Psi_t^{-t}(\Yvect)  = \Psi_t(\Yvect)$, it leads to
\begin{align}
\mathbb{E}_{\Yvect} \left[  \frac{\overline{\X}_t}{\Psi_t(\Yvect)} \overline{\X}_t \Psi_t(\Yvect)  \right] \underset{\alphavect \rightarrow \boldsymbol{0}}{=} \mathbb{E}_{\Yvect} \left[  \frac{\overline{\X}_t}{\Psi_t(\Yvect)} \Y_t \right]
\underset{\alphavect \rightarrow \boldsymbol{0}}{=} \mathbb{E}_{\Yvect} \left[  \frac{\Y_t}{\Psi_t^2(\Yvect)} \overline{\X}_t \Psi_t(\Yvect)\right]
\end{align}
By hypothesis, the memory functions satisfy Assumption~\ref{hyp:psi-alpha}.
Further remarking that the function $\Yvect \mapsto\Psi_t(\Yvect)$ satisfies Assumption~\ref{hyp:Poisson},  the autoregressive Poisson lemma~\ref{lem:Poisson-Stein} applies once again,  and using that $\Psi_t^{-t}(\Yvect)  = \Psi_t(\Yvect)$,  one gets
\begin{align}
\label{eq:third-term-est}
\mathbb{E}_{\Yvect} \left[  \frac{\overline{\X}_t}{\Psi_t(\Yvect)} \overline{\X}_t \Psi_t(\Yvect)  \right]&\underset{\alphavect \rightarrow \boldsymbol{0}}{=}  \mathbb{E}_{\Yvect} \left[ \left(  \frac{\Y_t}{\Psi_t^2(\Yvect)}\right)^{-t} \Y_t \right] \underset{\alphavect \rightarrow \boldsymbol{0}}{=}  \mathbb{E}_{\Yvect} \left[   \frac{(\Y_t - \alpha_t)\Y_t}{\Psi_t^2(\Yvect)} \right].
\end{align}
Equations \eqref{eq:second-term-est} and~\eqref{eq:third-term-est}, combined with the expansion provided in Equation~\eqref{eq:expand_est} demonstrates the asymptotic unbiasedness of $\mathsf{APURE}^{\mathcal{P}}$ stated in Equation~\eqref{eq:apure-est-thm1}.
\end{proof}

%\vspace{-1cm}

\section{Finite Difference Monte Carlo Estimators}
\label{sec:fdmc-proof}

\begin{proof}[Proof of Theorem~\ref{thm:fdmc}]
The Finite Difference Monte Carlo strategy applied to the prediction risk estimate $\mathsf{APURE}^{\mathcal{P}}$ of Equation~\eqref{eq:apure} consists in rewriting the second term in the definition of $\mathsf{APURE}^{\mathcal{P}}$ in Equation~\eqref{eq:apure} in a tractable way. The proof thus focuses on this term and aims at demonstrating that
\begin{align}
 \mathbb{E}_{\Yvect}  \left[ \sum_{t=1}^T \widehat{\X}_t^{-t} (\Yvect;\paramvect) \Psi_t(\Yvect) \Y_t\right] =    \mathbb{E}_{\Yvect,\mc} \left[ \sum_{t=1}^T  \widehat{\X}_t(\Yvect;\paramvect) \Psi_t(\Yvect)  \Y_t  \right.
 \quad - \left \langle   \mathrm{diag}(\boldsymbol{\alpha} \odot \Psivect(\Yvect) )  \partial_{\Yvect} \widehat{\Xvect} [\mc] , \mathrm{diag}(\Yvect)\mc \right\rangle \left], \textcolor{white}{\sum_{t=1}^T}\right.
\end{align}
which can be shown by alternatively demonstrating that for any $\paramvect\in\paramset$ and $t\in \lbrace 1, \hdots, T \rbrace$
\begin{align}
\label{eq:new-fdmc}
\mathbb{E}_{\Yvect} \left[ \widehat{\X}_t^{-t} (\Yvect;\paramvect) \Psi_t(\Yvect) \Y_t\right] = 
\mathbb{E}_{\Yvect,\mc} \left[   \widehat{\X}_t(\Yvect;\paramvect) \Psi_t(\Yvect)  \Y_t  -    \alpha_t  \Psi_t(\Yvect) \left( \partial_{\Yvect} \widehat{\Xvect} [\mc]\right)_t  \Y_t   \mcc_t  \right]
\end{align}
where $\partial_{\Yvect} \widehat{\Xvect} [\mc] \in \mathbb{R}^T$ is the differential of $\widehat{\Xvect}(\Yvect;\paramvect)$ with respect to the variable $\Yvect$ at $(\Yvect;\paramvect)$ applied to the $T$-dimensional Monte Carlo vector $\mc$ where for the sake of conciseness the point $(\Yvect;\paramvect)$ at which the differential is applied is omitted.\footnote{Remind that the differential of a function $f : \mathbb{R}^T \rightarrow \mathbb{R}^T$ at a given point $\Zvect \in \mathbb{R}^T$ is a linear application $\partial_{\yvect} f (\Zvect) [\boldsymbol{\cdot}] : \mathbb{R}^T \rightarrow \mathbb{R}^T$. }

To prove~\eqref{eq:new-fdmc}, first remark that, in the limit of small scale parameter $\alpha_t \rightarrow 0$,  Assumption~\ref{hyp:diff} implies that
\begin{align}
\label{eq:fd}
\widehat{\X}_t^{-t}(\Yvect;\paramvect) \underset{\alphavect\rightarrow \boldsymbol{0}}{=} \widehat{\X}_t(\Yvect) - \alpha_t \dfrac{\partial \widehat{\X}_t}{\partial \Y_t}.
\end{align}
The first term in the right-hand side of Equation~\eqref{eq:new-fdmc} stems directly from the first term in Equation~\eqref{eq:fd}. Then, since $\mc$ is a standard Gaussian vector, by definition $\mathbb{E}_{\mc} \left[ \mcc_s \mcc_t \right] = \delta_{s,t}$ where $\delta_{s,t}$ is the Kronecker delta,\footnote{By definition, $\forall s,t \in \mathbb{N}$, $\delta_{s,s} = 1$ and if $s\neq t$, $\delta_{s,t} = 0$.}
\begin{align}
\begin{split}
\label{eq:mc}
\dfrac{\partial \widehat{\X}_t}{\partial \Y_t}  =  \sum_{s = 1}^T \dfrac{\partial \widehat{\X}_t}{\partial \Y_s} \mathbb{E}_{\mc}\left[\mcc_s \mcc_t\right] = \mathbb{E}_{\mc}\left[ \sum_{s = 1}^T \dfrac{\partial \widehat{\X}_t}{\partial \Y_s} \mcc_s \mcc_t\right],
\end{split}
\end{align}
in which one recognizes the $t$th component of the differential of $\widehat{\Xvect}$ with respect to $\Yvect$ applied to the vector $\mc$
\begin{align*}
  \sum_{s = 1}^T \dfrac{\partial \widehat{\X}_t}{\partial \Y_s} \mcc_s =\left( \partial_{\Yvect} \widehat{\Xvect} [\mc]\right)_t.
\end{align*}
Then, multiplying Equations~\eqref{eq:fd} and~\eqref{eq:mc} by $\Psi_t(\Yvect)  \Y_t $ and combining them, one gets
\begin{align}
\begin{split}
\mathbb{E}_{\mc}\left[\widehat{\X}_t^{-t}(\Yvect;\paramvect)\Psi_t(\Yvect)  \Y_t \right] 
\underset{\alphavect\rightarrow \boldsymbol{0}}{=} \mathbb{E}_{\mc}\left[ \widehat{\X}_t(\Yvect)\Psi_t(\Yvect)  \Y_t - \alpha_t \left( \partial_{\Yvect} \widehat{\Xvect} [\mc]\right)_t \mcc_t \Psi_t(\Yvect)  \Y_t\right].
\end{split}
\end{align}
Taking the expectation with respect to $\mathbb{E}_{\Yvect}$ on both sides demonstrates~\eqref{eq:new-fdmc}, and finally shows the asymptotic unbiasedness of $\mathsf{APURE}^{\mathcal{P}}_{\mc}$ stated in Equation~\eqref{eq:fdmc-pred}.

Assuming that $\forall t\in \lbrace 1, \hdots, T\rbrace, \,  \forall \Yvect \in \mathbb{R}^T, \, \Psi_t(\Yvect) \neq 0$,  a similar proof holds for the estimation risk estimate.
Applying the Finite Difference Monte Carlo strategy to the estimation risk estimate $\mathsf{APURE}^{\mathcal{E}}$ of Equation~\eqref{eq:apure-est} amounts to rewrite the second term in the definition of $\mathsf{APURE}^{\mathcal{E}}$ in Equation~\eqref{eq:apure-est} in a tractable way. The proof thus focuses on this term and aims at demonstrating that
\begin{align}
 \mathbb{E}_{\Yvect} \left[ \sum_{t=1}^T \frac{\widehat{\X}_t^{-t} (\Yvect;\paramvect)}{ \Psi_t(\Yvect)} \Y_t\right] = \mathbb{E}_{\Yvect,\mc} \left[ \sum_{t=1}^T \sum_{t=1}^T \frac{ \widehat{\X}_t(\Yvect;\paramvect) }{ \Psi_t(\Yvect)}  \Y_t \right.\nonumber 
 \quad -\left \langle \mathrm{diag}(\alphavect \centerdot/ \Psivect(\Yvect)) \, \partial_{\Yvect} \widehat{\Xvect} [\mc],  \mathrm{diag}(\Yvect)\mc \right\rangle \left], \textcolor{white}{\sum_{t=1}^T}\right.
\end{align}
which can be shown by alternatively demonstrating that for any $\paramvect\in\paramset$ and $t\in \lbrace 1, \hdots, T \rbrace$
\begin{align}
\label{eq:new-fdmc-est}
 \mathbb{E}_{\Yvect} \left[\frac{ \widehat{\X}_t^{-t} (\Yvect;\paramvect) }{\Psi_t(\Yvect)} \Y_t\right] = 
 \mathbb{E}_{\Yvect,\mc} \left[ \frac{  \widehat{\X}_t(\Yvect;\paramvect) }{\Psi_t(\Yvect) } \Y_t  -   \frac{ \alpha_t }{\Psi_t(\Yvect) }\left( \partial_{\Yvect} \widehat{\Xvect} [\mc]\right)_t   \Y_t  \mcc_t  \right].
\end{align}

The first term in the right-hand of Equation~\eqref{eq:new-fdmc-est} stems directly from the first term in the Taylor expansion of Equation~\eqref{eq:fd}.
The second term Equation~\eqref{eq:fd} when injected in~\eqref{eq:new-fdmc-est} and rewritten leveraging the Monte Carlo strategy of Equation~\eqref{eq:mc} yields exactly the second term in Equation~\eqref{eq:new-fdmc-est}.
Taking the expectation with respect to $\mathbb{E}_{\Yvect}$ on both sides demonstrates~\eqref{eq:new-fdmc-est}, and finally shows the asymptotic unbiasedness of $\mathsf{APURE}^{\mathcal{E}}_{\mc}$ stated in Equation~\eqref{eq:fdmc-est}.
\end{proof}

\setlength{\bibsep}{4pt}

\bibliographystyle{plain}
\bibliography{biblio}

\end{document}